%% file: main.tex
\documentclass[11pt,letterpaper]{article}
\usepackage{fullpage}
\usepackage[round]{natbib}
\usepackage{amsmath,amsfonts,amssymb,amsthm,  mathrsfs,mathtools}
\usepackage[margin=1in]{geometry}
\usepackage{hyperref}
\hypersetup{
  colorlinks,
  citecolor=violet,
  linkcolor=blue,
  urlcolor=blue}
\usepackage{dsfont}
\usepackage{tcolorbox}
\numberwithin{equation}{section}   
\usepackage{booktabs} 
\usepackage[ruled,noend]{algorithm2e} 

\SetAlFnt{\small}
\SetAlCapFnt{\small}
\SetAlCapNameFnt{\small}
\SetAlCapHSkip{0pt}
\IncMargin{-\parindent}

\numberwithin{equation}{section}
\usepackage{macros}
\usepackage[suppress]{color-edits}
\addauthor{ss}{magenta} 
\usepackage{slivkins-theorems}

\title{Replication-proof Bandit Mechanism Design with Bayesian Agents}

  \newcommand{\country}[1]{#1.}
  \newcommand{\city}[1]{#1}
  \newcommand{\institution}[1]{#1}
  \newcommand{\email}[1]{Email: \texttt{#1}}
  \newcommand{\affiliation}{\thanks}
\author{
{Suho Shin
 \affiliation{
   \institution{University of Maryland}
   \city{College Park, MD}
   \country{USA}
 \email{suhoshin@umd.edu}
 }}
 \and
{Seyed Esmaeili
 \affiliation{
   \institution{Simons Laufer Mathematical Sciences Institute}
   \city{Berkeley, CA}
   \country{USA}
 \email{esmaeili@cs.umd.edu}
 }}
 \and
 {MohammadTaghi Hajiaghayi
 \affiliation{
   \institution{University of Maryland}
   \city{College Park, MD}
   \country{USA}
 \email{hajiagha@umd.edu }
 }}
}

\begin{document}
\maketitle
\begin{abstract}
\input{icml/icml-abstract}
\end{abstract}

\input{icml/icml-intro}
\input{icml/icml-model}
\input{icml/icml-single-agent}

\input{icml/icml-multi-agent}

\input{icml/icml-conclusion}

\bibliographystyle{plainnat}
\bibliography{ref}
\clearpage

\appendix
\input{icml/icml-appendix}

\end{document}

%% file: icml/icml-abstract.tex

We study the problem of designing replication-proof bandit mechanisms when agents strategically register or replicate their own arms to maximize their payoff.
Specifically, we consider Bayesian agents who only know the distribution from which their own arms' mean rewards are sampled, unlike the original setting of by Shin et al. 2022.
Interestingly, with Bayesian agents in stark contrast to the previous work, analyzing the replication-proofness of an algorithm becomes significantly complicated even in a single-agent setting.
We provide sufficient and necessary conditions for an algorithm to be replication-proof in the single-agent setting, and present an algorithm that satisfies these properties.
These results center around several analytical theorems that focus on \emph{comparing the expected regret of multiple bandit instances}, and therefore might be of independent interest since they have not been studied before to the best of our knowledge. 
We expand this result to the multi-agent setting, and provide a replication-proof algorithm for any problem instance.
We finalize our result by proving its sublinear regret upper bound which matches that of Shin et al. 2022.

%% file: icml/icml-intro.tex
\section{Introduction}\label{sec:intro}
Multi-armed bandit (MAB) algorithms are an important paradigm for interactive learning that focuses on quantifying the trade-off between exploration and exploitation \citep{slivkins2019introduction}, with various real-world applications such as recommender systems  \citep{barraza2020introduction,tang2014ensemble}, dynamic pricing by \citep{badanidiyuru2018bandits,badanidiyuru2014resourceful}, and clinical trials by \citep{djallel2019survey} to name a few.
They are used to model various practical applications where there is a set of actions (arms) to be selected (pulled) over a collection of rounds. 
Pulling an arm yields a reward sampled from its distribution which is generally different from the other arms' distributions. The objective is to maximize the sum of rewards accumulated through the rounds or equivalently minimize the \emph{regret} with respect to the optimal clairvoyant arm choice. 

In deploying MAB algorithms in real-world applications, however, the rewards are often not coming from stochastic distributions, but from a collection of rational individuals (agents) who seek to maximize their own payoff.
Consequently, there has been a surge of interest in bandit problem from the mechanism design perspective, which studies MAB algorithms that induce a good outcome with respect to agents' strategic behavior, while maintaining its efficiency.
Examples of such settings include two-sided matching markets, where agents learn their preferences through sequential interactions \citep{liu2021bandit,liu2020competing}, and learning in environments where agents can strategically modify the realized rewards of their arms \citep{braverman2019multi,feng2020intrinsic}.

Recently, \cite{shin2022multi} considered a problem where agents register arms to a platform that uses a MAB algorithm to minimize regret. 
Each agent has an original set of arms, and the mean reward of each arm is known to him a priori, \ie agents are \emph{fully-informed}.
Each agent receives a \emph{fixed constant portion} of the reward whenever its arm is pulled, and tries to maximize cumulative reward.
Importantly, each agent can potentially register the same arm \emph{more than once} (replication) at free cost, which possibly ruins the exploration-exploitation balance of the bandit algorithm.
Indeed, \cite{shin2022multi} show that ``ordinary'' algorithms such as upper-confidence-bound (UCB) induce rational agents to replicate infinitely many times, and thus the algorithm suffers linear regret.
To address this problem, \cite{shin2022multi} propose a hierarchical algorithm with two phases, and prove its truthfulness such that each agent's dominant strategy is not to replicate.

Although the introduction of replication-proof mechanism is remarkable, the assumption that each agent is fully informed about its own arms is theoretically strong and practically limited.
Intuitively, if an agent knows all of its own arms' mean rewards, there exists no reason for the agent to \emph{consider suboptimal arms}.
Indeed, this assumption significantly simplifies the equilibrium analysis so that the problem simply reduces to the case under which each agent only has a \emph{single arm}.
Then, it is immediate to check that any bandit algorithm is replication-proof under a single-agent setting.
Moreover, in practice, it is more reasonable to assume that each agent only has partial information regarding its own arms.
For instance, in online content platforms, each content provider usually does not know the quality of each content that has not been posted yet, but only can guess its outcome, based on the data possibly acquired from the previous contents he/she may have.

In this context, we study a \emph{Bayesian extension}\footnote{We remark that our notion of Bayesian differs from that of Bayesian mechanism in the literature. 
While the Bayesian mechanism computes the equilibrium by considering ex-interim payoff of the players in the incomplete information setting, \ie ex-interim incentive-compatible, our mechanism is DSIC in an ex-post manner..} of the replication-proof bandit mechanism design problem, where each agent only knows such a distribution from which its own arm's mean reward is sampled.\footnote{We refer to Appendix B for more discussion on the practical motivation of our problem setup.}
Given the prior distribution, each agent computes its ex-ante payoff by taking expectation over all the possible realization of its own arms, and commits to a strategy that maximizes the payoff.
Importantly, we aim to obtain a \emph{dominant strategy incentive-compatible} (DSIC) mechanism such that each agent has no incentive to deviate from its truthful strategy, \ie submitting all the original arms without any replication, regardless of the others' strategies.

While our extension itself is intuitive and simple, this brings \emph{significant challenges} in analyzing an algorithm's equilibrium compared to the fully-informed agent setting.
Notably, even in a single-agent case, it is not trivial to obtain a replication-proof algorithm.
This is in stark contrast to the fully-informed setting under which any bandit algorithm immediately satisfies replication-proofness.
The question of designing replication-proofness in the single-agent setting spawns an intriguing question of \emph{comparing expected regret of multiple bandit instances}, which will be elaborated shortly.

\paragraph{Outline of the paper}
In the following subsection, we present our main contributions. Due to the page limit, we defer our discussion on related works to Appendix A.
Then, we introduce the formal problem setup and preliminaries, and the failure of existing algorithms and its implication.
Next we develop our main intuition on constructing replication-proof algorithm by investigating the single-agent setting.
Correspondingly, we extend this result and provide our main result on  replication-proof algorithm with sublinear regret.
Practical motivations, discussion on an extension to asymmetric setting, details on the failure of H-UCB, and all the proofs are deferred to the appendix.
Due to the page limit, all the appendices can be found in the full paper~\cite{esmaeili2023replication}.

\subsection{Contributions and techniques}
Overall, our main results for replication-proof algorithms are presented in a step-by-step manner.
We start with investigating the previous algorithm H-UCB suggested by~\cite{shin2022multi}, which mainly builds upon standard UCB along with the idea of hierarchical structure.
Unlike from the fully-informed setting under which any bandit algorithm is replication-proof in the single-agent setting, we reveal that the standard UCB algorithm fails to be replication-proof even with a single-agent, regardless of the exploration parameter $f$, due to its deterministic (up to tie-break) and adaptive nature.
This phenomenon carries over to the multi-agent setting, so that UCB with hierarchical structure (H-UCB) fails to be replication-proof, regardless of the number of agents.

To tackle this challenge, we first thoroughly investigate the single-agent setting, and characterize necessary and sufficient conditions for an algorithm to be replication-proof.
Let us provide a simple example to explain the main technical challenges one would face in analyzing equilibriums.
\begin{example}
	Consider a single agent with two Bernoulli arms, say arm $a$ and $b$.
        The prior distribution is given by $Bern(0.5)$, \ie the expectation of both the arms' distributions are sampled from the prior distribution $Bern(0.5)$. 
        The agent receives a constant portion, say $40\%$ of the realized reward, whenever their arms are pulled.
	There exist four cases in terms of the pair of realized mean rewards $(\mu_a, \mu_b)$: (i) $(1,1)$,  (ii) $(1,0)$, (iii) $(0,1)$, and (iv) $(0,0)$, all of which occurs with probability $0.25$.
        Due to the nature of Bayesian agents, the agent cannot observe which arm has larger mean rewards, and suppose that the agent somehow commits to a strategy that only replicates the arm $a$, say the replicated arm $a'$.
	For cases (i) and (iv), the replicated arm has no effect since all the arms have the same mean rewards, \ie pulling any arm yields the same expected reward, so let's focus on case (ii) and (iii).
	In case (ii), the resulting bandit instance from the agent's strategy will be $\cI_A: (\mu_a, \mu_{a'},\mu_b) = (1,1,0)$ whereas in case (iii), it will be $\cI_B: (0,0,1)$
	Hence, compared to the strategy with no replication, the agent's gain from replicating arm $a$ is equal to the gain from the average of expected payoff in case (ii) and (iii), compared to the case in which the algorithm runs with bandit instance $\cI_O: (1,0)$ that does not have any replicated arm.
	Hence, if the agent has no incentive to replicate arm $a$ once, the expected regret under $\cI_O$ should be at most the average of that under $\cI_A$ and $\cI_B$.
\end{example}
%
Even under simplistic and standard algorithms such as UCB, it is not straightforward to answer the above question.
Intuitively, instance $\cI_A$ will incur smaller regret as the number of optimal arm increases, but the instance $\cI_B$ will have larger regret due to more suboptimal arms with mean reward $0$.\footnote{This can be shown using a coupling-like argument between random reward tapes of the arms, combined with a regret decomposition lemma. See~\cite{shin2022multi} for more details.}
Thus, the question is whether the loss from $\cI_B$ is larger than the gain from $\cI_A$ in average, compared to $\cI_O$.

More formally, we can refine the above observation as the following fundamental question in comparing the expected regret of multiple bandit instances.
\definecolor{mycolor}{rgb}{0.9, 0.90, 0.9}
\begin{tcolorbox}[colback=mycolor,colframe=gray!75!black,colframe=white]
\begin{center}
  \emph{Given a bandit instance, consider an adversary who uniformly randomly selects an arm and replicates it. Which bandit algorithms guarantee that the regret does not increase against the adversary?}
  \end{center}
\end{tcolorbox}
We affirmatively answer this question by proving that, there exists an algorithm that satisfies the above questions for \emph{any} bandit instance under some mild assumptions.
To this end, we introduce a notion of \emph{random-permutation regret}, which captures the expected regret of an agent's strategy with replication without observing the ex-post realization.
Then, we show that if an algorithm satisfies (i) \emph{truthfulness under random permutation} (TRP) such that the random-permutation regret does not increase with respect to replication, and (ii) permutation invariance such that the algorithm's choice does not depend on the arm indices, then the algorithm is replication-proof.
More surprisingly, we prove that these conditions are indeed \emph{necessary} for any algorithm to be replication-proof.\footnote{We refer to Section~\ref{sec:model} for precise definitions.}
In fact, the failure of UCB comes from violating TRP.
Meanwhile, we prove that exploration-then-commit (henceforth ETC) satisfies these properties, thus it is replication-proof in the single-agent setting.

For the multi-agent setting, we combine several techniques upon the observations above.
First, we borrow the hierarchical structure of~\cite{shin2022multi}, and analyze hierarchical algorithm with each phase running ETC, which effectively decomposes each agent's payoff from each other, thanks to its non-adaptive nature.
It remains as a major open problem whether it is necessary for the algorithm to have non-adaptive structure in order to satisfy TRP.
In addition, we introduce a novel restarting round to re-initialize the intra-agent statistics.
This largely simplifies the analysis of the selected agent in the exploitation rounds in the first phase, since the second phase simply reduces to the single-agent case with the highest empirical average reward.
We cement our result by proving that our algorithm has sublinear regret.

%% file: icml/icml-model.tex
\section{Model}\label{sec:model}
We consider a simultaneous game between a set of agents $[n] = \{1,2,\ldots, n\}$ and a principal.
Each agent $i \in \set{N}$ has a set of original arms $\set{O}_i = \{o_{i,1},\ldots, o_{i,l_i}\}$ given a constant $l_i \in \N$.
Each agent $i$ is equipped with a cumulative distribution $F_i$ from which each original arm $o_{i,k}$'s mean reward $\mu(o_{i,k})$ is sampled for $k \in [l_i]$.
Each arm $o_{i,k}$ further follows a reward distribution $G_{i,k}$ with the mean reward sampled as noted above.
We use $f_i$ and $g_{i,k}$ to denote the density function of $F_i$ and $G_{i,k}$,\footnote{We assume the existence of density function for ease of exposition, but all the results will easily be generalized.} and we assume that all the distributions are supported on $[0,1]$.
We further assume that $F_i$ is independent from $F_j$ for $i \neq j \in [n]$.
For instance, $F_i$ might be a uniform distribution over $[0,1]$, and each arm $a \in \cO_i$ has a Bernoulli reward distribution with the mean sampled from $F_i$.
We often deal with a \emph{single agent} setting with $n=1$.

\paragraph{Bayesian agents}
Importantly, we consider a \emph{Bayesian scenario} in which each agent is only aware of its own distribution $F_i$,\footnote{Discussion on the generalization to the asymmetric setting is presented in Appendix C.} but not the exact realization of $\mu_a$ for $a \in \set{O}_i$.
We note that such assumption is indeed practical since content providers (agents) in contents platform typically only have a partial information on how much their contents may attract the users, \eg based on his/her previous contents and outcomes therein.
This is in stark contrast to the setting of~\cite{shin2022multi} in which they assume \emph{fully-informed} agents such that any strategic agent does not have any incentive to register the arms except the best one reducing the problem into the single-arm case, which is not practical in many real-world scenarios.
We highlight that the equilibrium analysis in the fully-informed setting simply boils down to the case under which each agent only has a single original arm and decides how much to replicate that unique arm, whereas our extension significantly complicates the analysis of equilibrium, as will be elaborated shortly.

\paragraph{Agent's strategy}
Given the prior distribution $F_i$, the set of original $O_{i,k}$ and the principal's algorithm $\fA$, each agent $i$ decides how many times to replicate arm $k$, including $0$ of no replication, for each $k \in [l_i]$.
Precisely, agent $i$ commits to a set $\set{S}_i = \{s_{i,k}^{(c)}:c \in [c_{i,k}], k \in [l_i]\}$ in a strategic manner, where $c_{i,k}$ denotes the total number of replica for arm $o_{i,k}$.
The replicated arm $s_{i,k}^{(c)}$ has the same mean reward to the original arm, \ie $\mu(s_{i,k}^{(c)}) = \mu(o_{i,k})$ for any $i \in [n]$ and $k \in [l_i]$.
We often use $s_{i,k}^{(0)}$ to denote $o_{i,k}$ for simplicity.
Define $\cS = (\cS_1, \ldots, \cS_n)$ and $\cO = (\cO_1, \ldots, \cO_n)$.
Given all the agents' strategies $\set{S}$, the bandit algorithm $\fA$ runs and corresponding rewards are realized.

\paragraph{Mechainsm procedure}
The overall mechanism proceeds as follows: (i) The principal commits to a bandit algorithm (mechanism), (ii) Each agent $i \in \set{N}$  decides an action $\set{S}_i$ to register given $\mathfrak{A}$ to maximize own payoff, (iii) The mean rewards $\mu_a$ for each $a \in \set{O}$ are sampled, (iv) The bandit algorithm $\mathfrak{A}$ runs and rewards are realized, (v) The principal and the agent realize their own utility.\footnote{Note that if all the agents do not strategize and simply register $\set{O}_i$ in step (ii), our problem reduces to the standard stochastic multi-armed bandit problem.}

\paragraph{Agent's payoff}
Once an arm is selected by $\fA$, a fixed portion $\alpha \in (0,1)$ of the reward is paid to the agent who registers the arm.
We write $\muvec_i$ to denote $(\mu_{i,1}, \mu_{i,2}, \ldots, \mu_{i,l_i})$, \ie the vector of agent $i$'s (realized) mean rewards of the arms.
Let $E_{i}(\muvec_i)$ be an event that the mean reward of agent $i$'s arms are realized to be $\muvec_{i}$.
Recall that $E_i(\cdot)$ are mutually independent for for $i \in [n]$.
We often abuse $E_{i}$ to denote $E_i(\muvec_i)$ if the realization $\muvec_i$ is clear from the context.
We also write $\cS_{-i}$ to denote $(\cS_1, \ldots, \cS_{i-1}, \cS_{i+1}, \ldots, \cS_n)$, and $E(\muvec)$ for $\muvec = (\muvec_1,\muvec_2, \ldots, \muvec_n)$ to denote the event such that $\cap_{i \in [n]}E_i(\muvec_i)$.

We further introduce some notations given the principal's algorithm $\fA$.
Conditioned on the $E_{i}(\muvec_i)$, agent $i$'s \emph{ex-post payoff} of selecting strategy $\set{S}_i$ given the others' strategies is defined as
\begin{align*}
U_i(\set{S}_i, \set{S}_{-i} ;E(\muvec)) = \alpha \sum_{t=1}^T R_t\Ind{I^{(t)} \in \set{S}_i},
\end{align*}
where  $R_t$ refers to the reward at round $t$ and $I_t$ denotes the arm selected at round $t$.

Given the ex-post utility, each agent computes its \emph{ex-ante payoff} by marginalizing over the realization of its own arms' mean rewards, defined as follows
\begin{align*}
	u_i(\set{S}_i, \set{S}_{-i}) = \Ex{U_i(\set{S}_i, \set{S}_{-i};E(\muvec))},
\end{align*}
where the randomness in the expectation comes from $E(\muvec)$, i.e. the realization of the mean rewards, and possibly from the random bits of the algorithm $\fA$.
In words, each agent is rational \emph{expected utility maximizer} who maximizes its expected payoff.

\paragraph{Information and desiderata}
We now describe the information structure of the mechanism, and how each agent decides its own strategy based on its ex-ante payoff.
We first assume that the principal does not have any information regarding the agents' distributions nor the number of original arms, but can observe by whom each arm is registered.
This is often called \emph{prior-independent mechanism} in the mechanism design, \cf \cite{hartline2013bayesian}.
In addition, we assume that the agent does not have any information regarding the others.
Instead, our objective is to construct a \emph{incentive-compatible} mechanism such that each agent's \emph{dominant strategy} is to register the arms without any replication, without regard to what others may have proposed.
To this end, we first define a dominant strategy as follows.
\begin{definition}[Dominant Strategy]
Agent $i$'s strategy $\set{S}_i$ is dominant strategy if $u_i(\cS_i \cup \set{S}_{-i}) \ge  u_i(\set{S}'_i \cup \cS_{-i})$ for any $\set{S}'_i$ and $\set{S}_{-i}$.
\end{definition}
Correspondingly, we introduce the following notion of \emph{dominant strategy equilibrium}.
\begin{definition}[Equilibrium]
An algorithm $\mathfrak{A}$ has dominant strategy equilibrium	(DSE) if there exists a set of agent strategies $\set{S} = \{\set{S}_i\}_{i \in [n]}$ such that $\set{S}_i$ is dominant strategy for any agent $i \in [n]$.
\end{definition}
As per the definition, the total ordering of the agent's payoff with respect to its own strategy should remain the same regardless of the others' strategies $\cS_{-i}$.
Thus, we often omit the parameter $\cS_{-i}$ in denoting agent $i$'s payoff, \eg $u_i(\cS_i)$ instead of $u_i(\cS_i, \cS_{-i})$.
Note, however, that the realization of the other agents' mean rewards may affect the actual quantity of $u_i(\cS_i)$, although it does not affect which strategy agent $i$ will eventually commit to.
Since we are interested only in how each agent will "strategize", it does not lose any generality.

We finally introduce the following notion of truthfulness.
\begin{definition}[Replication-proof]
An algorithm $\mathfrak{A}$ is \emph{replication-proof} if it has a dominant strategy equilibrium $\{\set{O}_i\}_{i \in [n]}$ for any $T \in \N$.
\end{definition}
Replication-proofness is desirable as the agents are not motivated to replicate their own arms, thereby reducing the number of arms that a bandit algorithm faces.
Since the regret upper bound of a standard bandit algorithm usually depends on the number of arms, maintaining a smaller number of suboptimal arms in the system affects the eventual learning efficiency of the algorithm.
We often say that an algorithm is \emph{incentive-compatible} or \emph{truthful} to denote the replication-proofness~\citep{nisan2007algorithmic}.

On top of designing \emph{prior-independent} \emph{incentive-compatible} mechanism, we are also interested in a learning efficiency, \ie learning the reward distribution in a sample-efficient manner.
In the literature, this is usually described as cumulative regret, defined as follows.
\begin{definition}[Regret]
	Given time horizon $T$, strategy $\set{S}$ and the event $E(\muvec)$, the ex-post regret of algorithm $\mathfrak{A}$ is 
	\[\reg(\mathfrak{A}, T ; \set{S}, E(\muvec)) = \parans{\max_{a \in \cup_{i \in \cN} \cO_i}\sum_{t=1}^T \mu(a)} -\sum_{t=1}^T \mu(I_t),\]
	where $I_t$ denotes the arm selected by $\mathfrak{A}$ at round $t$.
\end{definition}
We emphasize that the optimal benchmark is selected from $\cup_{i \in \set{N}}\set{O}_i$, not from the registered set of arms $\cup_{i \in \set{S}}\set{S}_i$.
This benchmark is indeed important since the principal should incentivize the agents to decide an action that contains own best arm to maintain larger utility.
For example, if an algorithm motivates an agent to not register all the arms, then there might be a constant probability that the omitted arm has the largest mean reward, which essentially yields sublinear regret with respect to the regret defined above.
More formally, our notion of regret implicitly implies that any sublinear regret algorithm incentivizes all the agents to register all their original arms at least once.
We also remark that although the arms' mean rewards are sampled from prior distributions, which is often referred to Bayesian bandits (\eg see Ch.3 of \cite{slivkins2019introduction}), we deal with ex-post regret, \eg as does in \cite{agrawal2012analysis} which analyzes the ex-post regret of Thompson Sampling in Bayesian bandits.
Obviously, ex-post regret upper bound \emph{implies} Bayesian regret upper bound by simply taking an expectation.

\section{Failure of Existing Algorithms}
Before getting into our main results, the natural question one might ask is to analyze whether the existing algorithms by~\cite{shin2022multi} work.
They mainly propose hierarchical UCB (H-UCB), which is a variant of the standard UCB by~\cite{auer2002finite}.
The formal pseudocode of both UCB and H-UCB are deferred to the appendix.
While we provide the analysis with specific exploration parameter of $\ln t/n_a$, we note that the result easily carries over to arbitrary exploration parameters $f(t,n_a)$, once $f$ is increasing over $t$ and decreasing over $n_a$.
We omit further details as it is beyond our interest.

Mainly, H-UCB  consists of two phases: (i) in the first phase, it runs UCB1 by considering each agent as a single arm and selects one agent at each round, and (ii) in the second phase, it runs UCB within the selected agent.
Precisely, it maintains two types of statistics: (i) agent-wise empirical average reward $\muhat_i$, number of pulls per agent $n_i$, and (ii) arm-wise empirical average reward $\muhat_{i,a}$ and number of pulls per arm $n_{i,a}$.

\begin{theorem}\label{thm:negative}
	There exists a problem instance such that UCB1 is not replication-proof in the single-agent setting, and such that H-UCB  is not replication-proof for any number of agents.
\end{theorem}
The proof for the single-agent setting relies on a construction of bad problem instance which exploits the deterministic and adaptively exploring nature UCB.
To briefly explain the proof, we consider an agent with two Bernoulli arms whose prior distribution is supported on $\{0,1\}$, where it samples mean $0$ and $1$ equally likely.
Let's focus on a realization in which one arm has mean $0$, and the other has $1$.
Observe that replicating one arm essentially yields one of the following two bandit instances: (i) arms with mean rewards $(1,1,0)$ or (ii) arms with mean rewards $(1,0,0)$.
Here, we write $(x,y,z)$ to denote bandit instance with three arms  mean rewards $1,1,0$.
Note that, intuitively, (i) yields a lower expected regret, but higher in (ii), compared to the original instance with $(0,1)$.
Thus, our analysis mainly reduces to show that the loss from (ii) is larger than the gains from (i).
This can be made indeed true since after observing the dynamics of instance (ii), we can essentially set the time horizon to be exactly right \emph{before} the two arms with mean rewards zero are \emph{consecutively} chosen.
In this way, the loss from (ii) can be made smaller while the gains from (i) is moderately large, which implies the dominance of replicating strategy.
The proof for the multi-agent setting essentially is based on the fact that its Phase 2 does not guarantee replication-proofness within the chosen agent due to the single-agent argument, along with carefully chosen problem parameters to propagate this phenomenon to Phase 1 so as to increase the agent's expected utility under replication.

%% file: icml/icml-single-agent.tex
\section{Warm-up: Single-agent Setting}\label{sec:single}
As noted thus far, UCB is not replication-proof even in the single agent setting.
Indeed, as per Theorem~\ref{thm:negative}, it is not obvious to verify which algorithm would satisfy replication-proofness even in the single-agent setting.
In this section, we step towards in constructing replication-proof algorithm, by carefully investigating the single-agent setting first.
We present a set of conditions which is sufficient for an algorithm to be replication-proof in the single-agent setting, and then, as an example, will prove that exploration-then-commit (ETC) will satisfy these properties.
As we restrict our attention to single-agent setting in this section, we often omit the agent index $i$ in the notations.

To this end, we first introduce several notations.
For any natural number $l$, let $\cP_l$ be a set of all possible permutations $\sigma:[l] \mapsto [l]$.
Now we introduce a dictionary form of bandit instance $\cI$.
Given the realization of the arms' mean rewards, suppose that the set of mean rewards constitute a sorted sequence $\{1 \ge \mu_1 > \mu_2 > \ldots > \mu_l  > 0\}$, where some arms may share the same mean reward.
For each $\mu_a$ for $a \in [l]$, we count the number of arms which have mean reward $\mu_a$, and denote the count by $c_a$.
Then, our bandit instance $\cI$ can essentially presented as tuples of the mappings $\bigtimes_{a\in[l]}(\mu_a:c_a)$, which we say dictionary form of the bandit instance $\cI$.

Suppose that we have a standard multi-armed bandit instance $\cI$ with dictionary form $(\mu_a:c_a)_{a \in [l]}$.
Then, we define a \emph{permuted bandit instance} $\cI_\sigma$ to be a bandit instance with dictionary form of $\bigtimes_{a\in[l]}(\mu_a: c_{\sigma(a)})$, \ie arms with mean reward $\mu_i$ appears $c_{\sigma(a)}$ times.
For instance, given the bandit instance $\cI = (0.5: 1, 0.7: 2)$, consider a permutation $\sigma$ such that $\sigma(1) = 2$ and $\sigma(2)=1$ correspondingly.
Then, one can easily verify that $\cI_\sigma = (0.5:2, 0.7:1)$.

Now, we introduce the random permutation regret.
\begin{definition}[Random permutation regret]
	Given a single-agent bandit instance $\cI$ with $l$ arms, we define an algorithm $\fA$'s \emph{random permutation regret} (RP-Regret) as follows.
	\begin{align}
		\rpreg(\fA, T) 
		&= 
		\Exu{\sigma \in \cP_l}{\reg_{I_{\sigma}}(\fA, T)}
		\nonumber
		\\
		&=
		\frac{1}{|\cP_l|} \cdot \parans{\sum_{\sigma \in \cP_l}\reg_{I_{\sigma}}(\fA, T)}.\label{eq:epf}
	\end{align}
\end{definition}
Given the random permutation regret, we define the following property of algorithm, which will play an essential role in obtaining truthful algorithm in the single-agent setting.
\begin{definition}[Truthful under random permutation]
	Given an arbitrary single-agent bandit instance $\cI$ with $l$ arms,
	consider a truthful strategy $\cO$ and arbitrary strategy $\cS$.
	An algorithm $\fA$ is \emph{truthful under random permutation} (TRP) if
	\begin{align*}
		\rpreg(\fA,T;\cO) 
		\le 	
		\rpreg(\fA,T;\cS).
	\end{align*}
\end{definition}
Note that TRP requires the inequality to holds for arbitrary bandit instance $\cI$.
The foundation of these notions essentially builds upon the observation from our counterexample provided in Theorem~\ref{thm:negative}.
To elaborate more, 
consider a strategy $\cS$ that only replicates the first arm given a bandit instance $\mu_1\ge \mu_2 \ge \ldots \ge \mu_l$.
Then, we can observe that any permutation $\sigma \in \cP_l$ results in a bandit instance that belongs to the following family of bandit instances: define $\cI^* = \cup_{i \in [l]}\cI_i$ such that
\begin{align*}
    \cI_1 = \{\underbrace{\mu_1, \mu_1}_{repl}, \mu_2,  \ldots, \mu_l\} , 
    \ldots, \cI_l = \{\mu_1, \mu_2, \ldots, \underbrace{\mu_l,\mu_l}_{repl}\}.
\end{align*}


In words, $\cI_i$ denotes the case that the strategy $\cS$ replicates the arm with $i$-th highest reward.
For example, suppose that $\sigma \in \cP_l$ satisfies $\sigma(1) = i$ for some $i \in [l]$.
Since $\cS$ replicates the first arm, under the permutation $\sigma$, the arm with parameter $\mu_i$ is replicated, and thus it corresponds to $\cI_i$.
Then, our notion of TRP essentialy asks what is an algorithm that makes the expected regret of $\cI$ smaller than (or equal to) the average of expected regret of $\cI_1,\dots, \cI_l$.
Put it differently, TRP essentially asks the following fundamental question, as we pointed out in the introduction.
\begin{center}
	\emph{Given a bandit instance, does ``uniform randomly'' replicate one arm increase the expected regret?}
\end{center}
This foundation of TRP will play an essential role in constructing replication-proof algorithm in both the single-agent and the multi-agent setting.
Indeed, from our proof of Theorem~\ref{thm:negative}, one can observe that UCB does not satisfy the above question, and thus fails to be replication-proof.

\begin{definition}[Permutation invariance]
	An algorithm is \emph{permutation invariant} (PI) if given the random bits of the algorithm and the arms, the choice of which arms to pull remains the same (up to tie-breaking) when the index of the arms are permuted.
\end{definition}
Note that the permutation-invariance essentially requires the algorithm to be agnostic to the index of the arms, \ie does not use the arm index to choose the arm, possibly except the tie-breaking case.
One can easily verify that it holds for a broad class canonical algorithms such as UCB and $\eps$-greedy.

\begin{theorem}\label{thm:single_pf_rp}
	In the single-agent setting, if an algorithm is TRP and PI, then it is replication-proof.
\end{theorem}
The proof essentially follows from the observation that the algorithm's ex-ante regret can be partitioned into a disjoint set of random permutation regret, with respect to the ex-post realization of the mean rewards.
Based on this disjoint partition, we can effectively marginalize over the priors while maintaining the partition, and conclude that the dominance over the partition yields the overall dominance.

Furthermore, we observe that TRP is indeed a necessary condition for an algorithm to be replication-proof in the single-agent setting.
\begin{proposition}\label{thm:truthful-necessary}
	If an algorithm is not TRP, then it is not replication-proof.
\end{proposition}

\begin{discussion}
	One may wonder if PI can be violated for any replication-proof algorithm.
	For example, one can consider a black-box algorithm that first randomly selects an agent without using any statistical information but only with the agent indices, and runs a bandit algorithm within the agent.
	Further, suppose that it asymmetrically favors some agents' indices with higher probability in the first phase.
	Such an algorithm indeed is not permutation invariant, since the permutation of agent index would lead to a different outcome given the same reward tapes.
	From agents' perspective, however, each agent only needs to care about the secondary arm selection phase after the agent selection since they cannot change the probability to be selected in the agent selection phase.
	Using these observations, one can observe that such an algorithm is indeed replication-proof, but not permutation invariant.
	One may refine the notion of PI to be restricted within each agent, but we do not argue more details as it is beyond of our interest.
\end{discussion}

\paragraph{Constructing TRP algorithm}
We now present an algorithm that satisfies TRP.
In what follows, we mainly prove that exploration-then-commit (ETC) satisfies TRP once equipped with a proper parameter.
Its pseudocode is presented in Algorithm 3 in the appendix.
We suppose that the algorithm breaks tie in a uniform random manner.\footnote{Any tie-breaking rule does not hurt our analysis.}
ETC essentially decouples the exploration phase and the exploitation phase very explicitly, and thus brings a significant advantage in comparing its expected regret under several problem instances.

For analytical tractability, we pose some \emph{assumptions} on the support of the prior distributions.
Namely, for any agent $i$, $F_i$ has a discrete support over $[0,1]$ for $i \in \cN$.
Further, we define $\Delta_i$ be the minimum gap between any two possible outcomes from $F_i$, and let $\Delta = \min_{i \in \cN}\Delta$.
Furthermore, the algorithm knows this gap $\Delta$.
\begin{discussion}
	We discuss how one can weaken these assumptions at the end of this section, possibly at the cost of polylogarithmic blowup in the regret rate or another assumptions.
	We also remark that such assumptions, especially that the algorithm knows the minimum gap, often appear in the literature, \cf \cite{auer2002finite,audibert2010best,garivier2016explore}.
 \end{discussion}
 \begin{discussion}
	In practice, such scenario is fairly plausible in real-world applications. 
	For example in content platforms, suppose that a content creator comes up with several contents and is about to register. 
	Assume the reward of each content is just the number of views. 
	Typically, the content creator does not exactly know the reward of each content but has historical data on its previous contents. 
	The historical data can be used to structure a prior distribution of the quality of his/her contents. 
	Since each content creator usually has a very small number of original contents compared to the recommendation algorithm's time scale (one for each traffic), this prior distribution may consist of a fairly limited number of samples, thereby inducing a discrete support with a few number of points.
\end{discussion}

Then, our main result in the single-agent setting can be written as follows.
\begin{theorem}\label{thm:single_rp}
	ETC with exploration length $m \ge \frac{2}{\Delta^2} l \ln (2T)$ is TRP in the single-agent setting.
        Further, it has regret upper bound of
        \begin{align*}
            \sum_{a\in [l]} \Big( \frac{2\delta_a l\ln{(2T)}}{\Delta^2} + 1 \Big),
        \end{align*}
        where $l$ denotes the number of arms and $\delta_a$ denotes the gap of the mean rewards between optimal arm and arm $a$ for the single agent.
\end{theorem}
The proof essentially follows from the standard regret analysis of ETC along with carefully chosen $m$.
By setting $m$ sufficiently large, we can essentially decrease the probability that a suboptimal arm is chosen to be $o(1/T)$.
This will guarantee that the expected number of rounds the suboptimal arm is chosen is smaller than $1$, whereas any replicating strategy incurs a simple regret lower bound larger than this quantity.
Together with Theorem~\ref{thm:single_pf_rp} and the fact that ETC is PI (which is straightforward to verify), it follows that ETC is replication-proof in the single-agent setting.

Note that the proof of Theorem~\ref{thm:single_rp} heavily relies on the fact that the algorithm can observe $\Delta$ in advance and set the exploration length correspondingly.
Thus, despite we aim to obtain a \emph{prior-independent} algorithm that does not require any knowledge in the agent's prior distributions, it is not \emph{truly} prior-independent in an algorithmic manner since it requires some problem parameters to operate the algorithm.
We discuss several ways to refine such restrictions, thereby suggesting a road for truly prior-independent algorithm.
First, assume $\Delta$ is constant but not known to the algorithm.
In this case, one may replace $\Delta$ in the algorithm to be some increasing functions $f(T) = \omega(1)$, and it is straightforward to verify that Theorem~\ref{thm:single_rp} still holds, but in an asymptotical regime.
Similarly, one can effectively wipe out the dependency on the total number of original arms $l$ by assuming that $l = O(1)$ and replacing it with some increasing functions.
Furthermore, the knowledge on $T$ can be wiped out by the standard doubling trick, \cf see Chapter 1 in ~\cite{slivkins2019introduction}, without sacrificing the truthfulness.
Remark that all these refinements only introduce polylogarithmic blowup in the regret upper bound, where we do not provide a formal proof as it is a cumbersome application of standard techniques.
The following theorem spells out these arguments.
\begin{proposition}\label{thm:single-refine1}
	Assume $l = O(1),\Delta = \Omega(1)$.
	Then, there is a prior-independent algorithm with polylogarithmic regret satisfying TRP and PI in the single-agent setting for sufficiently large $T$, that does not require any information on problem parameters.
\end{proposition}

%% file: icml/icml-multi-agent.tex
\section{Multi-agent Replication-proof Algorithm}\label{sec:multi}

%
%
%
Now we turn our attention to the more general setting with multiple agents.
We start by using the machinery of \emph{hierarchical structure} in H-UCB, but in a different manner.
A direct implication of Theorem~\ref{thm:single_pf_rp} is that if we run a variant of H-UCB such that the first phase runs a simple uniform-random selection algorithm, and the second phase runs some arbitrary TRP and PI algorithms, then the resulting algorithm is replication-proof.
To formally see this why, let $\set{S}_i$ be the strategy of agent $i$.
Define $\Gamma_i(\set{S}_i)$ be the agent $i$'s utility under the single-agent setting with agent $i$, \ie when there are no other agents other than $i$.
Due to the uniformly random selection nature in the first phase, the expected utility of any agent $i$ can be simply written $u_i(\cS_i, \cS_{-i}) = {\Gamma_i(\set{S}_i)}/{n}$.
Note that the dynamic of the algorithm purely reduces to the dynamic of the second phase algorithm if there's only a single agent.
Hence, by Theorem~\ref{thm:single_pf_rp} and due to the TRP and PI of the second phase algorithm, we conclude that truthful registration is a dominant strategy for any agent $i$, and thus the result follows.
Note, however, that if we deploy uniform random selection algorithm in $\alg_1$, it essentially suffers \emph{linear regret} since it always choose any suboptimal agent with constant probability at every rounds.

To capture both of sublinear regret and replication-proofness, we need more sophisticated algorithm with sublinear regret in the first phase.
To this end, we present $\hbb$ presented in Algorithm 5 in the appendix, which adopts ETC in both phases along with additional device of \emph{restarting round}.
Similar to H-UCB, this consists of two phases each running ETC agent-wise and arm-wise manner, respectively.
For analytical tractability, we say that the set of arms belong to \emph{stochastically ordered family} if two arms $a$ and $b$ satisfy $\mu_a \ge \mu_b$, then the reward distribution of arm $a$ is (first-order) stochastic dominant over that of $b$, \ie $\Pr{r_a \ge x} \ge \Pr{r_b \ge x}$ for any $x\ge 0$, where $r_i$ denotes arm $i$'s reward random variable.\footnote{For example, Bernoulli arms and Gaussian arms parameterized by mean rewards, belongs to this family. We also refer to~\cite{yu2009stochastic} for more examples in exponential family.}
Further, define $L = \max_{i \in [n]}l_i$.
Then, our main result is as follows.
\begin{theorem}\label{thm:spf_rp}
	Consider a stochastically ordered family of arms, and discrete support of prior distributions.
	Consider $\hbb$ with $M \ge mL$ and $m = \frac{2}{\Delta^2} L\ln (2T)$, and $\tau = Mn$
	Then, $\hbb$ is replication-proof.\footnote{Similar to the discussion in Proposition~\ref{thm:single-refine1}, we can refine this results and make the algorithm truly prior-independent. We omit these details of beyond the interest.}
\end{theorem}
\begin{proof}[Proof sketch]
We provide a brief proof sketch.
The proof mainly relies on the single-agent result in Theorem~\ref{thm:single_pf_rp} along with a number of algebraic manipulations which heavily relies on the nature of ETC and the restarting round.
Let us focus on agent $i$.
We can decompose agent $i$'s total expected utility to be (i) exploitation phase utility and (ii) exploration phase utility.
Given a realization of the reward tapes, the utility from (i) and (ii) is independent due to the nature of ETC.
Since (i) simply reduces to running single-agent ETC for each agent, we can compare (i) of two strategies using the single-agent result, thus truthful strategy is at least better for (i) thanks to the single-agent theorem.

For (ii), if agent $i$ is not selected in the exploitation phase, it is simply zero, so we can only focus on the case when agent $i$ is selected in the exploitation phase.
Thus, its expected utility from (ii) can be written as multiplication of (a) probability that agent $i$ has the largest empirical average reward and (b) the expected utility by being pulled over the exploitation rounds.
For (b), since our choice of $\tau$ \emph{restarts} Phase 2 for any agent exactly at the beginning of exploitation phase of Phase 1, the expected utility therein simply reduces to the expected utility of single-agent setting with agent $i$ given the time horizon $T-Mn$.
Thus, again by our single-agent result of Theorem~\ref{thm:single_pf_rp}, the truthful strategy yields larger (or equal) quantity for (b) as well.
Finally, by carefully using a series of concentration inequalities, we bound the difference from (a) between the truthful strategy and any other strategy is relatively smaller and by doing so, loss from (a) cannot make up the gains from (b) and (i), which eventually implies the dominance of truthful strategy.
\end{proof}


We finalize our results by presenting regret upper bound of the $\hbb$ with ETC and ETC,  which concludes that it achieves sublinear regret as well as replication-proofness.
\begin{theorem}\label{thm:regret}
	Set $\alg_1$ be ETC with $M = \max(mL, \sqrt{T\ln T})$, and $\alg_2$ be ETC with $m =\nicefrac{2}{\Delta^2}\cdot L\ln(2T)$.
	Then, $\hbb$ has expected regret of $O(\frac{nL^3\sqrt{T \ln T}}{\Delta^3})$.
\end{theorem}
The regret analysis is based on a construction of clean event on each agent under which the agent's empirical average reward becomes close enough to his own optimal arm's mean reward.
The deviation probability of this clean event can be obtained using tail bounds along with the regret analysis of ETC, which essentially implies the necessity of having $\alg_1$'s exploration length polynomial.
The regret bound then easily follows from obtaining the probability that all the clean event holds.
The formal proof is deferred to appendix.\footnote{We remark that the algorithm's dependency on the problem parameters can be weakened similar to Proposition~\ref{thm:single-refine1}, but we omit the details as it goes beyond our interest.}



%% file: icml/icml-conclusion.tex
\section{Conclusion}
We study bandit mechanism design problem to disincentivize replication of arms, which is the first to study Bayesian extension of \cite{shin2022multi}.
Our extension brings significant challenges in analyzing equilibrium, even in the single-agent setting.
We first prove that H-UCB of \cite{shin2022multi} does not work.
We then characterize sufficient conditions for an algorithm to be replication-proof in the single-agent setting.
Based on the single-agent result, we obtain the existence of replication-proof algorithm with sublinear regret, by introducing a restarting round and exploiting the structural property of exploration-then-commit algorithm.
We further provide a regret analysis of our replication-proof algorithm, which matches the regret of H-UCB.


%% file: icml/icml-appendix.tex
\input{icml/icml-ucb}


\begin{proof}[Proof of Theorem~\ref{thm:single_pf_rp}]
	Let $\fA$ be an algorithm with symmetry and TRP.
	Given a single agent problem instance $\cI$ with original set $\cO$ with $l$ arms, let $\cS$ be a truthful strategy that registers $\cO$.
	Consider any non-truthful strategy $\cS$ such that it replicates $r_i \in \Z_0$ times for arm $a$ for $a \in [l]$, where $\Z_0 =  \{0,1,2,\ldots\}$ denotes the set of nonnegative integers.
	Note that since $\cS$ is non-truthful strategy, at least one $r_a$ has to be strictly positive for some $a \in [l]$.
	We abuse $r_a=0$ to say that arm $a$ is only registered once without any replication, and write $r = (r_1,\ldots, r_l)$.
	For $v = (v_1,\ldots, v_l) \in \R_{\ge 0}^l$, we write $g(v)$ to denote the density of the mean rewards of arms having $v$, \ie $g(v) = \prod_{a=1}^l f(v_a)$, where $\R_{\ge 0}^l$ denotes the set of $l$-dimensional nonnegative real vectors.
	For simplicity, we assume that the prior distribution of the agent has probability density function, and does not have any point mass.\footnote{This assumption is for ease of exposition, and does not restrict any of our results. For instance, if there exists point mass so that two arms may realize the same mean reward with strictly positive probability, then we can essentially unionize those two arms and consider it as a single arm. In this way, one can verify that all the analysis carry over.}
	We further write $dv$ to denote $dv_1 dv_2 \ldots dv_l$.
	We use $\reg_{v, r}$ to denote the regret of problem instance such that the arms' mean rewards are realized to be $v$ and each arm $i$ is replicated $r_i$ times.
	We finally write $\sigma(v)$ to the permuted vector given a permutation $\sigma \in S^l$, and $\sigma(v)_i$ to denote $i$-coordinate of $\sigma(v)$.
	
	
	Observe that the expected utility in playing the strategy $\cS$ can be written as follows.
	\begin{align*}
		u(\cS) 
		&= 
		\int_{\mu_1 = v_1}\ldots \int_{\mu_l = v_l} \Ex{U(\cS)\ |\ \bigtimes_{i \in [l]}(v_a: r_a)}g(v)dv
		\\
		&=
		\int_{\mu_1 = v_1} \int_{\mu_2 = v_2 < v_1}\ldots \int_{\mu_k = v_l< v_{l-1}} 
		\sum_{\sigma \in \cP_l}\Ex{U(\cS)\ |\ \bigtimes_{i \in [l]}(v_{\sigma(a)}: r_a)}g(v)dv
		\\
		&=
		\int_{\mu_1 = v_1} \int_{\mu_2 = v_2 < v_1}\ldots \int_{\mu_k = v_l< v_{l-1}} 
		\sum_{\sigma \in \cP_l}\Ex{v_1T - \reg_{\sigma(v),r}(T) \ |\ \bigtimes_{i \in [l]}(v_a: r_a)}g(v)dv,
	\end{align*}
	where in the last equation, we condition on $\bigtimes_{i \in [n]}(v_i:r_i)$ since $\reg_{\sigma(v),r}(T)$ already accounts for the permuted bandit instance.
	Further, due to the permutation invariance of the algorithm, applying $\sigma^{-1}$ to all the arm indices does not change the expected regret, and thus we can further expand it to be
	\begin{align}
		u(\cS)
		&=
		\int_{\mu_1 = v_1} \int_{\mu_2 = v_2 <v_1}\ldots \int_{\mu_l = v_l< v_{l-1}} 
		\sum_{\sigma \in \cP_l}\Ex{v_1T - \reg_{v,\sigma(r)}(T) \ |\ \bigtimes_{a \in [l]}(v_a: r_a)}g(v)dv.
		\label{ineq:05131256}
	\end{align}
	Let $\vec{0}_l$ denote the $l$-dimensional zero vector.
	We now prove the following claim.
	\begin{claim}
		For any $v \in R_{\ge 0}^l$, $\sigma \in \cP_l$, and $r \in \Z_{\ge 0}^l$, we have
		\begin{align*}
			\sum_{\sigma \in \cP_l}\Ex{\reg_{v,\sigma(\vec{0}_l)}(T)} 
			\le 
			\sum_{\sigma \in \cP_l}
			\Ex{\reg_{v,\sigma(r)}(T)}.
		\end{align*}
	\end{claim}
	\begin{proof}
		This directly follows from our definition of TRP.
		More formally,
		\begin{align*}
			\sum_{\sigma \in \cP_l}\Ex{\reg_{v,\sigma(\vec{0}_l)}(T)} 
			\tag{def. of RP-Regtret}
			&=
			\frac{\rpreg(\fA,T;\cO)}{|\cP_l|}
			\\
			&\le
			\frac{\rpreg(\fA,T;\cS) }{|\cP_l|}
			\tag{since $\fA$ is TRP}
			\\
			&=
			\sum_{\sigma \in \cP_l}
			\Ex{\reg_{v,\sigma(r)}(T)}.
			\tag{def. of RP-Regret}
		\end{align*}
		It completes the proof.
	\end{proof}
	By the claim, we can further expand from \eqref{ineq:05131256} as follows.
	\begin{align*}
		u(\cS)
		&\le
		\int_{\mu_1 = v_1} \int_{\mu_2 = v_2 < v_1}\ldots \int_{\mu_l = v_l< v_{l-1}} 
		\sum_{\sigma \in \cP_l}\Ex{v_1T - \reg_{v,\sigma(\vec{0}_l)}(T) \ |\ \bigtimes_{a \in [l]}(v_a: r_a)}g(v)dv,	
	\end{align*}

	Since $\reg_{v, \sigma(\vec{0}_l)}$ is equivalent to the regret of the truthful strategy $\cO$ given the ex-post realization $v$, we obtain
	\begin{align*}
		u(\cS)
		&\le
		\int_{\mu_1 = v_1} \int_{\mu_2 = v_2 < v_1}\ldots \int_{\mu_l = v_l< v_{l-1}} 
		\sum_{\sigma \in \cP_l}\Ex{v_1T - \reg_{\cI}(T) \ |\ \bigtimes_{a \in [l]}(v_a: r_a)}g(v)dv
		\\
		&=
		\int_{\mu_1 = v_1} \int_{\mu_2 = v_2 }\ldots \int_{\mu_l = v_l} 
		\Ex{U(\cO)\ |\ \bigtimes_{a \in [l]}(v_a: r_a)}g(v)dv
		\\
		&=
		u(\set{O}),
	\end{align*}
	and thus the truthful registration $\set{O}$ is a dominant strategy.
\end{proof}


\section{Proof of Proposition~\ref{thm:truthful-necessary} and Theorem~\ref{thm:single_rp}}
Here, we provide the proofs for the necessity of TRP for truthfulness in the single-agent setting, and the main result in the single-agent setting.
\begin{proof}[Proof of Proposition~\ref{thm:truthful-necessary}]
	Consider a single-agent setting.
	Since the algorithm $\fA$ does not satisfy TRP, there exists a problem instance with original arms $\cO$ and the agent's strategy $\cS$, time horizon $T$ such that 
	\begin{align*}
		\rpreg(\fA,T;\cO) 
		\le 	
		\rpreg(\fA,T;\cS).
	\end{align*}
	As observed in the proof of Theorem~\ref{thm:single_pf_rp}, the expected utility of the strategies $\cO$ and $\cS$ can be compared using $\rpreg(\fA,T;\cO)$ and $\rpreg(\fA,T;\cS)$, so the result follows.
    \sscomment{double check}
\end{proof}

Before presenting the proof of Theorem~\ref{thm:single_rp}, we provide the pseudocode of the standard ETC algorithm to make the paper self-contained.
\begin{algorithm}
\textbf{Input}: exploration length $m$, tie-breaking rule\\
Play each arm for $m$ rounds.\\
Play the arm with the highest empirical mean for all remaining rounds
\caption{\textsc{ETC}}\label{alg:ETC}
\end{algorithm}

\begin{proof}[Proof of Theorem~\ref{thm:single_rp}]
	Recall that we will deal with a single agent.
	Consider a bandit instance $\cI$ in which there exists $l$ arms with parameters $v_1 \ge v_2 \ge \ldots \ge v_l$, and the time horizon is $T$.
	Let $n_a(t)$ be the number of rounds that arm $a$ is pulled and $\muhat_a(t)$ be the empirical average reward, at the end of round $t$, and define $\delta_a = v_1 - v_a$ for $a \in [l]$.
	
	Our proof mainly relies on the following lemma.
	\begin{lemma}\label{th:ETC_main_th}
		Suppose that the realization of the arms's mean rewards is given, and let $a^*$ be the arm the largest mean reward.
		If $m \ge \frac{2 l}{\Delta^2} \ln{2T}$, then for a sub-optimal arm we have:
		\begin{align*}
			\Ex{n_a(T)} \leq m+1.
		\end{align*}
		Especially for $m = \frac{2}{\Delta^2} l \ln (2T)$, the expected regret is upper bounded by:
		\begin{align*}
    		\reg(T) 
    		\leq 
    		\sum_{a\in [l]} (m+1) \Delta_i 
    		\leq 
    		\sum_{a\in [l]} \Big( \frac{2\delta_a l\ln{(2T)}}{\Delta^2} + 1 \Big).
		\end{align*}	
	\end{lemma}
	\begin{proof}[Proof of the lemma]
		Due to the structural nature of ETC, for any sub-optimal arm $a$ we have
		\begin{align*}
    		\Ex{n_a(T)} \leq m + (T-mk) \Pr{\hat{\mu}_a(m) \ge \hat{\mu}_{a^*}(m)} 
    	\end{align*}
    	Now we focus on upper bounding $\Pr{\hat{\mu}_a(m) \ge \hat{\mu}_{a^*}(m)}$.
    	\begin{claim}\label{cl:concen}
    		For any arm $a \neq a^*$, the following holds.
    		\begin{align*}
    			\Pr{\hat{\mu}_a(m) \ge \hat{\mu}_{a^*}(m)} 
    			\le 
    			2\exp(-\frac{\delta_a^2 m}{2})
    		\end{align*}
    	\end{claim}
    	\begin{proof}[Proof of the claim]
    		We can expand the inequality as follows.
    		\begin{align*}
    			\Pr{\hat{\mu}_a(m) \ge \hat{\mu}_{a^*}(m)} &= \Pr{(\hat{\mu}_a(m) -v_a) + (v_{a^*} - \hat{\mu}_{a^*}(m)) \ge \delta_a}
    			\\ 
    	    	&\leq 
    	    	\Pr{(\hat{\mu}_a(m) -v_a) \ge \frac{\delta_a}{2}} + \Pr{(v_{a^*} - \hat{\mu}_{a^*}(m)) \ge \frac{\delta_a}{2}}
    	    	\\ 
    	    	&\leq 
    	    	2 \exp(- \frac{\delta_a^2}{4} \frac{1}{2} \frac{1}{\frac{1}{4m}}) \tag{Hoeffding's inequality} \\
    	    	& \leq 2\exp(- \frac{\delta_a^2 m}{2}),
			\end{align*}
			which finishes the proof.
			Note that we use the following well-known type of Hoeffding's inequality.
			\begin{theorem}
				Let $X_1,\ldots, X_n$ be independent random variables such that $a_i \le X_i \le b_i$ almost surely.
				Define $S_n = X_1 + \ldots X_n$.
				For any $t >0$,
				\begin{align*}
					\Pr{|S_n - \Ex{S_n}| \ge t}
					\le 
					2\exp\parans{-\frac{2t^2}{\sum_{i=1}^n (b_i - a_i)^2}}.	
				\end{align*}
			\end{theorem}
    	\end{proof}
		Using the claim above, we can further expand,
		\begin{align*}
       		\Ex{n_a(T)} 
       		&\leq
       		m + (T-mk) \Pr{\hat{\mu}_a(m) \ge \hat{\mu}_{a^*}(m)} 
       		\\
       		&\leq 
       		m + 2T\exp(- \frac{\delta_a^2 m}{2})
       		\tag{Claim~\ref{cl:concen} $\&$ $T-mk \le T$}
      		\\
       		&\leq
       		m + 2T  \exp(- \Big(\frac{\delta_a}{\Delta}\Big)^2 \ln{2T})
       		\tag{$m\ge 2/\Delta^2 \ln(2T)$}
       		\\
       		&\leq 
       		m+ 2T  \exp(- \ln{2T})
       		\tag{Since $\frac{\delta_a}{\Delta} \ge 1$} 
       		\\
       		&=
       		m + 1.
		\end{align*}
		This proves the upper bound on $\Ex{n_a(T)}$.
		To obtain the upper bound on the regret, one can use the following well-known regret decomposition lemma.
		\begin{lemma}[Regret decomposition]
			The regret can be writtten as
			$\reg(T) = \sum_{a \in [l]} \Ex{n_i(T)\delta_a}$.
		\end{lemma}

		From this, the regret bound easily follows as below
		\begin{align*}
    		\reg(T)
    		&\leq 
    		\sum_{a\in [l]} (m+1) \delta_a 
    		\\
    		&= 
    		\sum_{a\in [l]} (\frac{2l}{\Delta^2} \ln{2T}+1) \delta_a 
    		\\ 
    		&\leq 
    		\sum_{a\in [l]} ( \frac{2l\delta_a}{\Delta^2} \ln{2T} +1).
    		\tag{Since $\Delta_a\leq 1$} 
		\end{align*}
		This completes the proof.
	\end{proof}
	
	Now, consider a strategy $\cS \neq \cO$ of which the replication vector is $r = (r_1,\ldots, r_l)$ (as defined in the proof of Theorem~\ref{thm:single_pf_rp}), and at least one coordinate being strictly positive.
	Consider a realization $v \in \R_{\ge 0}^l$, and a permutation $\sigma(v)$.
	Let $\cI_{\sigma}$ be a bandit instance in which the arms are permuted by $\sigma$ from $\cI$.

	Given $v_1 \ge \ldots \ge v_l$, let $\cI_{\sigma,v}$ be the instance in which $i$-th arm has mean reward $v_{\sigma(a)}$ for $a \in [l]$, \ie $\cI_{\sigma, v} = \bigtimes_{a \in [l]}(v_{\sigma(a)}: r_a)$.
	Denote the expected regret under $\cI_{\sigma,v}$ by $\reg_{I_{\sigma,v}}$.
	
	From the regret decomposition lemma, we have
	\begin{align*}
		\reg_{I_{\sigma,v}}(T) = \sum_{a\in [l]}	\Ex{n_{\sigma(a),t}}\delta_{\sigma(a)} \parans{r_a + 1}.
	\end{align*}
	Due to the initial exploration phase of ETC, obviously $\Ex{n_{\sigma(a),t}} \ge m$.
	Hence we obtain
	\begin{align*}
		\reg_{I_{\sigma,v}}(T) 
		&\ge
		\sum_{a\in [l]} m \cdot \delta_{\sigma(a)}(r_a + 1).
	\end{align*}
	Now, RHS in the random permutation regret \eqref{eq:epf} can be written as follows.
	\begin{align}
		\sum_{\sigma \in \cP_l}\reg_{\cI_{\sigma,v}}(T)/|\cP_l|
		&\geq
		\frac{\sum_{\sigma \in \cP_l}\sum_{a\in [l]} m \cdot \delta_{\sigma(a)}(r_a + 1)}{l!}
		\nonumber
		\\
		&\geq
		\frac{\sum_{a\in [l]} \sum_{\sigma \in \cP_l} m \cdot \delta_{\sigma(a)}(r_a + 1)}{l!}
		\nonumber
		\tag{Change order of summation}
		\\
		&=
		\frac{\sum_{a\in [l]} \sum_{\sigma \in \cP_l} m \cdot \delta_{a}(r_{\sigma(a)} + 1)}{l!}
		\nonumber
		\tag{Sum over $\sigma$ equals that over $\sigma^{-1}$}
		\\
		&=
		\frac{\sum_{a\in [l]} m \cdot \delta_a(l-1)!\sum_{j=1}^l (r_{j}+1)}{l!}
		\nonumber
		\tag{fix $\sigma(a) = j$ and sum over $\sigma \in \cP_l$	}
		\\
		&=
		\frac{\sum_{a\in [l]} m \cdot \delta_a (\norm{r}_1+l)}{l}\label{ineq:05090221}
	\end{align}
	Due to our construction of $r$, we have $\norm{r}_1 \ge 1$, and finally we obtain
	\begin{align}
		\reg_\cI(T) 
		&\le
		\sum_{a\in [l]} (m+1) \delta_a 
		\label{ineq:05090222}
		\\
		&\le
		\sum_{a\in [l]} m\delta_a \frac{l+1}{l},
		\nonumber
		\tag{$m = \frac{2l}{\Delta^2}\ln 2T \ge l$}
		\\
		&\le
		\frac{\sum_{a\in [l]} m \cdot \delta_a\sum_{j=1}^l (r_{j}+1)}{l}
		\tag{$\norm{r}_1 \ge 1$}
		\nonumber
		\\
		&\le
		\sum_{\sigma \in \cP_l}\reg_{\cI_{\sigma,v}}(T)/|\cP_l|,
		\nonumber
	\end{align}
	where the last inequality follows from the inequality~\eqref{ineq:05090221}.
	Note further that $\reg_\cI(T)$ is $\rpreg$ of truthful strategy $\cO$ given the realization $v$, and RHS denotes that of a strategy $\cS$ whose replication vector is $r$.
	Thus,
	\begin{align*}
		\rpreg(\fA, T, \cO) \le \rpreg(\fA, T, \cS),	
	\end{align*}
	for $\fA$ being ETC, and it concludes that ETC is TRP.
\end{proof}

\section{Proof Sketch of Proposition~\ref{thm:single-refine1}}\label{apd:thm:single-refine1}
We here provide the proof sketch of Proposition~\ref{thm:single-refine1}, where we omit the details as it is a cumbersome application of standard techniques.
The formal pseudocode is provided in Algorithm~\ref{alg:pid-ETC}.
\begin{algorithm}
	\textbf{Input}: Tie-break rule $\tau$\\
	$T_{prev} \gets 0$\\
	\For{$i=1,2,\ldots$}
	{
		$T \gets 2^i$\\
		\While{$t\le T$}
		{
			$m \gets \floor{\ln^4(T)}$
			Play each arm $m$ rounds\\
			Play the arm with the highest empirical mean for all remaining rounds, given $\tau$
		}
		$T_{prev} \gets T$
	}
	\caption{Prior-indepdenent ETC}\label{alg:pid-ETC}
\end{algorithm}
	
	To prove the proposition, let us fix some $T$ and focus on the rounds that operates based on this $T$.
	Due to the assumptions, there exist sufficiently large $T$ such that $l,1/\Delta \le\ln T$.
	Then, for sufficiently large $T$, it follows that
	\begin{align}
		m = \floor{\ln^4(T)} \ge l \ln T \Delta^2.\label{eq:msatis}
	\end{align}
	If $T$ is not sufficiently large, the while loop will break before gets to the exploitation phase, and therefore it does not change the expected payoff of the agent.
	On the other hand, if $T$ is sufficiently large, then since $m$ satisfies \eqref{eq:msatis}, by Theorem~\ref{thm:single_rp}, each partition will be replication-proof for fixed $T$.
	Therefore, iterating over $T = 2^{1},2^2,\ldots$ preserves the replication-proofness of the algorithm, and thus the overall algorithm is replication-proof.
	Note further that the exploration length $\ln^4(T)$ only amounts for polylogarithmic factor to the regret of the algorithm.
\begin{algorithm} 
\SetAlgoLined
\textbf{Input}: Tie-breaking rule, agent set $\cN$, arm set $\cS_i$ for $i \in \cN$, restarting round $\tau \in [T]$\\
Let $\muhat_i \gets 0$, $n_i\gets 0$, $\muhat_{i,a} \gets 0$, $n_{i,a} \gets 0$ for $i \in \cN$ and $a \in \cS_i$\\

 \For{$t=1,2,\ldots,T$}{
  \If{$t = \tau+1$}
  {
  $\muhat_{i,a} \gets 0; ~ n_{i,a} \gets 0$ for $\forall a \in \cS_i$, $\forall i \in \cN$
  }
  \If{$n_i < M$ for some $i \in \cN$}
  {
  	$\hat{i} \gets i$
  }
  \Else{$\hat{i} \gets \argmax_{i \in \cN}\muhat_{i}$}
  
  \If{$n_{\hat{i},a} < m$ for some $a \in \cS_{\hat{i}}$}
  {
  	$\hat{a} \gets a$
  }
  \Else{$\hat{a} \gets \argmax_{a \in \cS_{\hat{i}}}\muhat_{\hat{i},a}$}  
  Pull $\hat{a}$ of agent $\hat{i}$ and obtain reward $R_t$\\
  Update statistics:
  $\muhat_{\hat{i},\hat{a}} \leftarrow \tfrac{\muhat_{\hat{i},\hat{a}}n_{\hat{i}, \hat{a}} + R_t}{n_{\hat{i}, \hat{a}}+1}; ~
  \muhat_{\hat{i}} \leftarrow \tfrac{\muhat_{\hat{i}}n_{\hat{i}} + R_t}{n_{\hat{i}}+1};~ n_{\hat{i}, \hat{a}} \leftarrow n_{\hat{i}, \hat{a}} + 1; ~~ n_{\hat{i}} \leftarrow n_{\hat{i}} + 1;$
 }
 \caption{Hierarchical ETC with Restarting ($\hbb$)\label{alg:hbb}}
\end{algorithm}

\section{Proof of Theorem~\ref{thm:spf_rp}}
\begin{proof}
	\sscomment{proofread again, and structurize}
	First, note that our $\hbb$ is well-defined since our choice of $M$ and $m$ ensures that the exploitation phase of the first phase starts after the exploitation of the second phase ends.
	For simplicity, denote the first phase ETC by $\alg_1$ and the second phase ETC by $\alg_2$.
	Pick an agent $i$ with arms $\{o_{i,1}, \ldots, o_{i,l}\}$, let $\cS_i$ be agent $i$'s arbitrary strategy which replicates arm $o_{i,j}$ of $r_{i,j} \in \N_0$ times for each $j \in [l_i]$.
	Let $r_i = (r_{i,1},\ldots, r_{i,l_i})$, and denote by $v_i = (v_{i,1}, \ldots, v_{i,l_i})$ the realization of agent $i$'s original arms' mean rewards.
	Fix strategies of all the other agents $j \neq i$ by $\cS_{-i}$, and fix a realization $v_{-i}$ of all the other agents' arms' mean rewards.
	Henceforth, we will omit $\cS_{-i}$ and $v_{-i}$ from the arguments in the notations presented hereafter, unless confusion arises.
	In addition, we assume that all the prior distributions do not have point mass for simplicity, but it can easily be generalized as discussed in the proof of Theorem~\ref{thm:single_pf_rp}.
	
	Let $i^*$ denote the agent selected at the exploitation phase of $\alg_1$, and $l_i$ to denote agent $i$'s number of arms.
	We further write $\muhat_i(M)$ to denote the empirical average reward of agent $i$ when agent $i$ is selected $M$ times,
	and $\muhat_{i,a}(m)$ to denote the empirical average reward of agent $i$'s arm $a$ when it's selected $m$ times.
	Then, given $v_i$ and $r_i$, recall that the bandit instance agent $i$ submits can be presented as $\bigtimes_{i \in [l]}(v_a: r_a)$.
	
	Now, similar to the proof of Theorem~\ref{thm:single_pf_rp}, we first write agent $i$'s expected utility as follows.
	\begin{align}
		u_i(\cS_i)
		=
		\int_{\mu_{i,1} = v_{i,1}} \int_{\mu_{i,2} = v_{i,2} < v_{i,1}}\ldots \int_{\mu_{i,l_i} = v_{i,l_i}< v_{i,l_i-1}} 
		\sum_{\sigma \in \cP_{l_i}}\Ex{U_i(\cS_i)\ |\ \bigtimes_{a \in [l_i]}(v_{i,\sigma(a)}: r_{i,a})}g(v_i)dv_i.\label{eq:04262152}
	\end{align}
	
	We first characterize the expected utility of agent $i$ based on its utility under the single agent setting with agent $i$.
	To this end, we introduce several notations.
	Write $\Gamma_i^{\cS_i}(t, \sigma(v_i))$ to denote agent $i$'s utility when he is selected $t$ times under the single agent setting in $\alg_2$ when its arms' mean rewards are realized to be $\sigma(v_i)$, and he plays strategy $\cS_i$ with replication vector $r_i$.
	Given $M$, $m$, and $\cS_i$, define an event
	\[
	C^{\cS_i}_i(y,\sigma(v)) = \{\max_{\cN \ni j\neq i}\parans{\muhat_j(M) }= y \ | \ \bigtimes_{a \in [l_i]}(v_{i, \sigma(a)}:r_{i,a})\},
	\]
	\ie the event that agents other than $i$ have maximal agent-wise empirical average reward $y$ at the end of the exploration phase in $\alg_1$.
	Again given $M$, $m$, and $\cS_i$, we further define following two quantities:
	\begin{align*}
		p^{\cS_i}_i(y, \sigma(v_i)) &= \CPr{C_i(y)}{\bigtimes_{a \in [l_i]}(v_{i, \sigma(a)}:r_{i,a})}
		\\
		q^{\cS_i}_i(y,\sigma(v_i)) &= \CPr{\muhat_i(M) \ge y}{\bigtimes_{a \in [l_i]}(v_{i, \sigma(a)}:r_{i,a})}.
	\end{align*}
	Remark that all the quantities above are conditioned on the event that $v_{-i}$ is realized and the other agents are playing $\cS_{-i}$.

	The following claim essentially characterizes the expected utility of agent $i$ under strategy $\cS$.
	\begin{claim}
		Agent $i$'s expected utility can be written as follows.
		\begin{align}
			\CEx{U_i(\cS_i)}{\bigtimes_{a \in [l_i]}(v_{i, \sigma(a)}:r_{i,a}}
			&=
			\Ex{\Gamma_i^{\cS_i}(M,\sigma(v_i))}  \nonumber
			\\
			&~~~+ \Ex{\Gamma_i^{\cS_i}(T-Mn,\sigma(v_i))} \int_0^1p_i^{\cS_i}(y,\sigma(v_i))q_i^{\cS_i}(y,\sigma(v_i))dy.\label{eq:04271231}	
		\end{align}
	\end{claim}
	\begin{proof}[Proof of the claim]
		Since $\hbb$ is running with the first phase being ETC, the expected utility of agent $i$ can be decomposed as (i) expected utility in the exploration phase and (ii) expected utility in the exploitation phase.
		Note that (i) essentially coincides with the expectation of $\Gamma_i^{\cS_i}(M, \sigma(v_i))$ since it chooses agent $i$ exactly $M$ times regardless of the random rewards.
		Moreover, due to our restarting at time $\tau = Mn$, when the exploitation phase of $\alg_1$ starts, all the parameters of $\alg_2$ is initialized.
		Thus, conditioned on the event that agent $i$ has the largest empirical average reward at round $\tau = Mn$, its expected utility in the exploitation phase again coincides with the single agent case expected utility $\Gamma_i^{\cS_i}(T-Mn, \sigma(v_i))$.
		Otherwise, if agent $i$ has not been selected as the empirically best agent, then its conditional expected utility is simply zero.
		Note that the probability of agent $i$ being the empirically best agent can be written as follows,\footnote{To be precise, we need to take account for the tie-breaking rule and the event that tie happens, but we omit the detail as it can be easily handled and is beyond of our interest.}
		\begin{align*}
			\Pr{\text{agent $i$ is the empirically best agent at round $\tau = Mn$}}
			&=
			\int_{0}^1 \Pr{\muhat_i(M) \ge y}\Pr{C_i(y)}
			\\
			&=
			\int_0^1p_i^{\cS_i}(y,\sigma(v_i))q_i^{\cS_i}(y,\sigma(v_i))dy,
		\end{align*}
		and we finish the proof of the claim.
	\end{proof}
	Using the claim above, for any strategy $\cS_i$, we can further expand its conditional expected utility $\sum_{\sigma \in \cP_{l_i}} \CEx{u_i(\cS_i)}{\bigtimes_{a \in [l_i]}(v_{i, \sigma(a)}:r_{i,a})}$ as follows.
	\begin{align}
		\sum_{\sigma \in \cP_{l_i}}
		& 
		\CEx{U_i(\cS_i)}{\bigtimes_{a \in [l_i]}(v_{i, \sigma(a)}:r_{i,a})}\nonumber
		\\
		&=
		\sum_{\sigma \in \cP_{l_i}}
		\parans{
		\Ex{\Gamma_i^{\cS_i}(M,\sigma(v_i))} + \Ex{\Gamma_i^{\cS_i}(T-Mn,\sigma(v_i))} \int_0^1p_i^{\cS_i}(y,\sigma(v))q_i^{\cS_i}(y,\sigma(v_i))dy
		}
		\nonumber
		\\
		&=
		\parans{
		\sum_{\sigma \in \cP_{l_i}}
		\Ex{\Gamma_i^{\cS}(M,\sigma(v_i))}
		}
		+
		\sum_{\sigma \in \cP_{l_i}}
		\parans{
		\Ex{\Gamma_i^{\cS_i}(T-Mn,\sigma(v_i))}
		\int_{0}^1p_i^{\cS_i}(y,\sigma(v_i))q_i^{\cS_i}(y,\sigma(v_i))dy
		}
		\nonumber
		\\
		&=
		\parans{
		\sum_{\sigma \in \cP_{l_i}}
		\Ex{\Gamma_i^{\cS_i}(M,\sigma(v_i))}
		}
		+
		\int_{0}^1
		\parans{\sum_{\sigma \in \cP_{l_i}}
		\Ex{\Gamma_i^{\cS_i}(T-Mn,\sigma(v_i))}
		p_i^{\cS_i}(y,\sigma(v_i))q_i^{\cS_i}(y,\sigma(v_i))}dy
		\nonumber
		\\
		&=
		\parans{
		\sum_{\sigma \in \cP_{l_i}}
		\Ex{\Gamma_i^{\cS_i}(M,\sigma(v_i))}
		}
		+
		\int_{0}^1
		p_i^{\cS_i}(y,\sigma(v_i))
		\parans{
		\sum_{\sigma \in \cP_{l_i}}
		\Ex{\Gamma_i^{\cS_i}(T-Mn,\sigma(v_i))}
		q_i^{\cS_i}(y,\sigma(v_i))
		}dy,\label{eq:04271244}
	\end{align}
	where the third equality follows from Fubini's theorem, and the last inequality holds since $p_i^{\cS_i}(y,\sigma(v_i))$ is independent from $\sigma$.
	
	In case of the truthful strategy $\cO_i$, since any permutation $\sigma$ yields the same problem instance, all the terms inside \eqref{eq:04271244} remains the same.
	Thus,
	\begin{align}
		\sum_{\sigma \in \cP_{l_i}} &
		\CEx{U_i(\cO_i)}{\bigtimes_{a \in [l_i]}(v_{i, \sigma(a)}:r_{i,a})}\nonumber
		\\
		&~~~~~~=
		|\cP_{l_i}| 
		\parans{
		\Ex{\Gamma_i^{\cO_i}(M,
		v_i)}
		+
		\Ex{\Gamma_i^{\cO_i}(T-Mn,v_i)}
		\int_{0}^1p_i^{\cO_i}(y,v_i)q_i^{\cO_i}(y,v_i)dy
		}.\label{eq:05082325}
	\end{align}
	Now it suffices to prove that the RHS of \eqref{eq:05082325} is at least that of \eqref{eq:04271244}.
	
	We first observe that the first term can effectively compared based on our result on the single-agent setting.
	\begin{claim}\label{cl:reduc_to_single}
		For any strategy $\cS_i$, we have
		\begin{align*}
			\sum_{\sigma \in \cP_{l_i}}\Ex{\Gamma_i^{\cO_i}(M, \sigma(v_i))}	
			\ge
			\sum_{\sigma \in \cP_{l_i}}\Ex{\Gamma_i^{\cS_i}(M, \sigma(v_i))}.
		\end{align*}
	\end{claim}
	\begin{proof}
		Note that $\Gamma_i^{\cO_i}(M, \sigma(v_i))$ denotes the agent $i$'s utility at round $M$ under the single agent setting.
		Thus, the claim directly follows from Theorem~\ref{thm:single_rp} since we run ETC with $m = 2l/\Delta^2 \ln (2T)$ as Phase 2 algorithm.
	\end{proof}
	Thus, it remains to compare the latter term that involves the integral.
	Given $v_i \in \R_{\ge 0}^{l_i}$, for $a \in [l_i]$, write $Y_{i,a}$ to denote the event that the arm with mean $v_a$ of agent $i$ is selected in the exploitation rounds of $\alg_2$ for agent $i$, \ie having the largest empirical average reward at the end of the exploration rounds in $\alg_2$ for agent $i$.
	
	\sscomment{Proofread}
	Further, similar to the definition of $\Gamma_i^{\cS_i}(\cdot, \cdot)$, we analogously define $\Gamma_{i,a}^{\cS_i}(\cdot, \cdot)$ to be the random cumulative rewards when there exists only agent $i$ and only arm $a$ under $\alg_2$.
	That is, $\Gamma_{i,a}^{\cS_i}(t, \sigma(v_i))$ is just a summation of $t$ random variables where each random variable follows the agent $i$'s arm's reward distribution parameterized by the mean reward $\sigma(v_i)$.
	Note that our construction in equation~\eqref{eq:04262152} verifies $v_{i,1} \ge v_{i,2} \ge \ldots \ge v_{i,l_i}$.
	Define $\delta_{i,a} = v_{i,1} - v_{i,a}$ for $a \in [l_i]$.
	
	We now construct a lower bound of $q_i^{\cO_i}(y,v_i)$ in the following. \sscomment{$o_i$ and $o$ are both used. Also, make consistent $\sigma(v)$ or $v$.}
	\begin{claim}\label{cl:tail_lower}
		For any $y \in [0,1]$, and $\sigma \in \cP_{l_i}$, consider arbitrary $v_i \in [0,1]^{l_i}$ such that $v_{i,1} \ge v_{i,2} \ge \ldots \ge v_{i,l_i}$.
		 Then, we have
		\begin{align*}
			q_i^{\cO_i}(y, v_i)
			\ge
			\Pr{\Gamma_{i,1}^{\cO_i}(M,v_i) \ge My} \parans{1 - 2\sum_{a=2}^{l_i} \exp(-\frac{\delta_{i,a}^2 m }{2})}
		\end{align*}

	\end{claim}
	\begin{proof}[Proof of the claim]
		Due to the independence of rewards from arm $a$ at exploitation phase and the random variable $\Ind{Y_{i,a}}$, we can rewrite $q_i^{\cO}(y,\sigma(v))$ as follows.
		\sscomment{FIXING...}
		\begin{align}
			q_i^{\cO_i}(y,v_i)
			&=
			\Pr{
			\sum_{a=1}^{l_i}
			\parans{
			\Gamma_{i,a}^{\cO_i}(m,v)
			+
			\Ind{Y_{i,a}}\Gamma_{i,a}^{\cO_i}(M-l_im, v_i)
			}
			\ge My
			}
			\nonumber
			\\
			&=
			\sum_{a=1}^{l_i}
			\CPr{
			\sum_{a=1}^{l_i} 
			\parans{
			\Gamma_{i,a}^\cO(m,v)
			+
			\Ind{Y_{i,a}}\Gamma_{i,a}^{\cO}(M-l_im, v)
			}
			\ge My
			}{Y_{i,a}}\Pr{Y_{i,a}}
			\nonumber
			\\
			&=
			\sum_{a=1}^{l_i}
			\CPr{
			\parans{
			\sum_{a=1}^{l_i} 
			\Gamma_{i,a}^\cO(m,v)
			}
			+
			\Gamma_{i,a}^{\cO}(M-l_im, v)
			\ge My
			}{Y_{i,a}}\Pr{Y_{i,a}}
			\nonumber
			\\
			&\ge
			\CPr{
			\parans{
			\sum_{a=1}^{l_i} 
			\Gamma_{i,1}^\cO(m,v)
			}
			+
			\Gamma_{i,1}^{\cO}(M-l_im, v)
			\ge My
			}{Y_{i,1}}\Pr{Y_{i,1}}
			\nonumber
		\end{align}
		
		Note further that due to the positive correlation between the event $Y_{i,1}$ and the tail event, it follows that
		\begin{align*}
			\CPr{
			\parans{
			\sum_{a=1}^{l_i} 
			\Gamma_{i,1}^{\cO_i}(m,v)
			}
			+
			\Gamma_{i,1}^{\cO_i}(M-l_im, v)
			\ge My
			}{Y_{i,1}}
			\ge
			\Pr{
			\parans{
			\sum_{a=1}^{l_i} 
			\Gamma_{i,1}^{\cO_i}(m,v)
			}
			+
			\Gamma_{i,1}^{\cO_i}(M-l_im, v)
			\ge My
			}.
		\end{align*}
		
		Let us write
		\begin{align*}
			\psi_i(M, y, v) 
			= \Pr{
			\parans{
			\sum_{a=1}^{l_i} 
			\Gamma_{i,1}^{\cO_i}(m,v)
			}
			+
			\Gamma_{i,1}^{\cO_i}(M-l_im, v)
			\ge My
			},
		\end{align*}
		then we have
		\begin{align}
			q_i^{\cO_i}(y,v_i) \ge \psi_i(M, y,v)\Pr{Y_{i,1}}
			\label{ineq:05090246}
		\end{align}
		
		Since $Z_{i,1}$ is the complement of $\cup_{a=2}^l Z_{i,a}$, observe that
		\begin{align*}
			\Pr{Z_{i,1}^c}
			&=
			\Pr{\cup_{a=2}^l Z_{i,a}}
			\\
			&\le
			\sum_{a=2}^l\Pr{Z_{i,a}}
			\tag{Union bound}
			\\
			&\le
			\sum_{a=2}^l\Pr{\muhat_{i,1}(m) \le \muhat_{1,a}(m)}
			\tag{Definition of $Z_{i,a}$}
			\\
			&\le
			\sum_{a=2}^l 2\exp(-\frac{\delta_{i,a}^2 m}{2}).
			\tag{Claim~\ref{cl:concen}}
		\end{align*}
		
		Hence, we obtain
		\begin{align*}
			q_i^{\cO}(y,\sigma(v))
			\ge
			\psi_i(M, y, v)  \parans{1 - 2\sum_{a=2}^l \exp(-\frac{\delta_{i,a}^2 m }{2})},
		\end{align*}
		and it proves the claim.
	\end{proof}

	\begin{claim}\label{cl:tail_upper}
		For any strategy $\cS_i$, $\sigma \in S_{l_i}$, $y \in [0,1]$ and $v \in [0,1]^{l_i}$, we have 
		\begin{align*}
			q_i^{\cS_i}(y, \sigma(v))	 
			\le 
			\psi_i(M, y, v)
		\end{align*}
	\end{claim}
	\begin{proof}[Proof of the claim]
		This easily follows from the fact that in the probability $\psi(M,y,v)$, the exploration phase runs with the original $l_i$ arms for $m$ times each and then the exploration phase runs with the optimal arm $(i,1)$, whereas in the probability $q_i^{\cS_i}(y, \sigma(v))$, the exploration runs with the original $l_i$ arms upon the replicated arm for $m$ times each, and then exploitation runs with the empirically best arm.
		That is, we can couple the first $l_i \cdot m$ rounds for the exploration phase of the original $l_i$ arms, and then the exploitation is dominant for $\psi(M,y,v)$ due to the stochastic dominance of the arm $(i,1)$ over the other arms.
		
	\end{proof}

	Now, combining Claim~\ref{cl:tail_lower} and~\ref{cl:tail_upper}, we obtain
	\begin{align}
		q_i^{\cO}(y, \sigma(v)) - q_i^{\cS}(y,\sigma(v))
		&\ge
		-2
		\psi_i(M, y, v)
		\parans{\sum_{a=2}^{l_i} \exp(-\frac{\delta_{i,a}^2 m }{2})}
		\nonumber
		\\
		&\ge
		-2l_i 
		\psi_i(M, y, v)
		\parans{\exp(-\frac{\Delta^2 m }{2})}
		\tag{Due to $\delta_{i,a} \ge \Delta$}
		\\
		&\ge
		-2l_i \cdot \frac{q_i^{\cO}(y, v)}{\Pr{Y_{i,1}}} \cdot
		\parans{\exp(-\frac{\Delta^2 m }{2})}\tag{Due to ineq. \eqref{ineq:05090246}}
		\\
		&\ge
		-2l_i \cdot \frac{q_i^{\cO}(y, v)}{1 - \sum_{a=2}^{l_i} 2\exp(-\frac{\delta_i^2 m}{2})} \cdot
		\parans{\exp(-\frac{\Delta^2 m }{2})}
		\nonumber
		\\
		&\ge
		-2l_i
		\cdot 
		\frac{q_i^{\cO}(y, v)}{1 - 2l_i\exp(-\frac{\Delta^2 m}{2})} \cdot
		\parans{\exp(-\frac{\Delta^2 m }{2})},
		\label{ineq:05090213}
	\end{align}
	where the fourth inequality follows from the concentration on $Z_{i,1}$ in the proof of Claim~\ref{cl:tail_lower}.

	Now we consider two cases regarding $\sigma$: (i) when $\sigma(1) = 1$ and (ii) otherwise.
	
%
	
	Recall that it suffices to prove that the RHS of \eqref{eq:05082325} is at least that of \eqref{eq:04271244}.
	Due to Claim~\ref{cl:reduc_to_single}, it suffices to prove the following inequality.
	\begin{align}
		\int_{0}^1
		&p_i^{\cS}(y,\sigma(v))
		\parans{
		\sum_{\sigma \in S_l}\underbrace{\paranm{\Ex{\Gamma_i^{\cO}(T-Mn,\sigma(v))}
		q_i^{\cO}(y,\sigma(v))
		-
		\Ex{\Gamma_i^{\cS}(T-Mn,\sigma(v))}
		q_i^{\cS}(y,\sigma(v))}}_{\Theta(\sigma, v)}
		}dy 
		\ge 
		0.
		\label{eq:05090205}
	\end{align}
%
%
%
%
%

	Note that the term inside the summation can be decomposed as follows.
	\begin{align*}
		\Theta(\sigma, v)
		=
		\paranm{\Ex{\Gamma_i^{\cO}(T-Mn,\sigma(v))} - \Ex{\Gamma_i^{\cS}(T-Mn,\sigma(v))}}q_i^{\cO}(y,\sigma(v))
		\\
		+
		\Ex{\Gamma_i^{\cS}(T-Mn,\sigma(v))}
		\paranm{q_i^\cO(y,\sigma(v)) - q_i^{\cS}(y,\sigma(v))}.
	\end{align*}
	
	Due to \eqref{ineq:05090213}, we further obtain
	\begin{align*}
		\Theta(\sigma, v) 
		&\ge 
		\underbrace{\paranm{\Ex{\Gamma_i^{\cO}(T-Mn,\sigma(v))} - \Ex{\Gamma_i^{\cS}(T-Mn,\sigma(v))}}q_i^{\cO}(y,\sigma(v))}_{\Theta_1(\sigma, v)}
		\\
		&-
		\underbrace{2\Ex{\Gamma_i^{\cS}(T-Mn,\sigma(v))}
		\cdot l_i\cdot \frac{q_i^{\cO}(y, v)}{1 - 2l_i\exp(-\frac{\Delta^2 m}{2})} \cdot
		\parans{\exp(-\frac{\Delta^2 m }{2})}}_{\Theta_{2}(\sigma, v)} 
		.
	\end{align*}
	By summing over $\sigma \in S_{l_i}$, for $\Theta_1(\sigma, v)$, we have
	\begin{align*}
		\sum_{\sigma \in S_l} \Theta_1(\sigma, v)	
		&=
		\sum_{\sigma \in S_l}\paranm{\Ex{\Gamma_i^{\cO}(T-Mn,\sigma(v))} - \Ex{\Gamma_i^{\cS}(T-Mn,\sigma(v))}}q_i^{\cO}(y,\sigma(v))
		\\
		&=
		q_i^{\cO}(y,\sigma(v)) \cdot \sum_{\sigma \in S_{l_i}}\paranm{\Ex{\Gamma_i^{\cO}(T-Mn,\sigma(v))} - \Ex{\Gamma_i^{\cS}(T-Mn,\sigma(v))}}
	\end{align*}
	
	From inequality~\eqref{ineq:05090221} and \eqref{ineq:05090222} we know that
	\begin{align*}
		\sum_{\sigma \in S_{l_i}}\paranm{\Ex{\Gamma_i^{\cO}(T-Mn,\sigma(v))} - \Ex{\Gamma_i^{\cS}(T-Mn,\sigma(v))}}
		&\ge
		\sum_{a \in [l_i]} \delta_{i,a} (\frac{m}{l_i} - 1)
		\\
		&\ge
		\sum_{a \in [l_i]} \delta_{i,a} (\frac{2l_i \ln 2T / \Delta^2}{l_i} - 1)
		\\
		&= 
		\sum_{a \in [l_i]} \delta_{i,a} (\frac{\ln 2T}{\Delta^2} - 1),
	\end{align*}
	which implies
	\begin{align*}
		\sum_{\sigma \in S_l} \Theta_1(\sigma, v)
		&\ge
		q_i^{\cO}(y,\sigma(v)) 
		\cdot \sum_{i \in [l]} \delta_{i,a} (\frac{\ln 2T}{\Delta^2} - 1).
	\end{align*}
	
	Putting $m = 4 \cdot \frac{2}{\Delta^2} (\max_i l_i) \ln(2T)$, observe that the following holds.
	\begin{align*}
		\exp(-\frac{\Delta^2 m}{2}) \le \frac{1}{4l_iT}	
	\end{align*}

	Hence, for $\Theta_2(\sigma,v)$, we further obtain
	\begin{align*}
		\sum_{\sigma \in S_{l_i}}\Theta_2(\sigma,v)
		&=
		2\Ex{\Gamma_i^{\cS}(T-Mn,\sigma(v))}
		\cdot l_i\cdot \frac{q_i^{\cO}(y, v)}{1 - 2l_i\exp(-\frac{\Delta^2 m}{2})} \cdot
		\parans{\exp(-\frac{\Delta^2 m }{2})}
		\\
		&\le
		4\Ex{\Gamma_i^{\cS}(T-Mn,\sigma(v))}
		\cdot l_i\cdot q_i^\cO(y,v)\cdot \frac{1}{4l_iT}
		\\
		&=
		\Ex{\Gamma_i^{\cS}(T-Mn,\sigma(v))}
		 q_i^\cO(y,v)\frac{1}{4T}
		\\
		&\le
		\frac{q_i^{\cO}(y,v)}{4}\tag{Since $\Ex{\Gamma_i^{\cS}(T-Mn,\sigma(v))} \le T$}
	\end{align*}
	
	Combining the above observations, for $\Theta(\sigma, v)$, we obtain
	\begin{align*}
		\sum_{\sigma \in S_{l_i}}\Theta(\sigma, v)
		&\ge
		q_i^{\cO}(y,v)
		\parans{
		\sum_{a \in [l_i]} \delta_{i,a} (\frac{\ln 2T}{\Delta^2} - 1)
		-
		\frac{1}{4}
		}
		\\
		&\ge
		q_i^{\cO}(y,v)
		\parans{
		\sum_{a \in [l_i]}  (\frac{\ln 2T}{\Delta} - 1)
		-
		\frac{1}{4}
		}
		\\
		&=
		q_i^{\cO}(y,v)\parans{l_i \cdot \ln(2T) - l_i -\frac{1}{4}} 
		\\
		&\ge 0,
	\end{align*}
	where the last inequality follows from $T \ge 1$ and $q_i^\cO(y,v) \ge 0$.
	Thus, inequality \eqref{eq:05090205} is true, and it finishes the proof.

\end{proof}

\section{Proof of Theorem~\ref{thm:regret}}\label{apd:thm:regret}

\begin{proof}
	Suppose without loss of generality, that the arms within each agent is sorted in decreasing order, \ie for any $i$, $\mu_{i,1} \ge \mu_{i,2} \ge \ldots \mu_{i,l_i}$.
	Assume further that the agents are also sorted in decreasing order with respect to their best arm, \ie $\mu_{1,1} \ge \mu_{2,1} \ldots \ge \mu_{n,1}$.
	For $i \in [n]$, define $\Delta_i = \mu_{1,1} - \mu_{i,1}$, and for each $a \in [l_i]$, let $\delta_{i,a} = \mu_{i,1} - \mu_{i,a}$.
	Let $\muhat_i(M)$ be agent $i$'s empirical average reward at the end of round $M$.
	
	Let $I_t$ be the arm selected at round $t$, and $r_t$ be the corresponding reward.
	We define an intra-agent regret $\reg_i(t)$ to be the regret occurred in Phase 2 until when agent $i$ is selected exactly $t$ times.
	Precisely,
	\begin{align*}
		\reg_i(t)
		=
		t\mu_{i,1}-\sum_{\tau=1}^tr_{\tau}\Ind{I_\tau \in \cS_i}
	\end{align*}
	
	\begin{claim}
		For any agent $i$ and for any $x \in [0,1]$, we have
		\begin{align*}
			\Pr{|\muhat_{i}(M) - \mu_{i,1}| \ge x}
			\le
			\frac{1}{Mx}	\sum_{a=1}^{l_i}\Big( \frac{2 L^2\ln{(2M)}}{\Delta^2} + 1 \Big)		
		\end{align*}
		
	\end{claim}
	\begin{proof}[Proof of the claim]
		\begin{align}
		\Pr{|\muhat_{i}(M) - \mu_{i,1}| \ge x}
		&=
		\Pr{|M\muhat_i(M) -M\mu_{i,1}| \ge Mx}
		\label{eq:05101204}
	\end{align}
	We observe that $|M\muhat_i(M) - M\mu_{i,1}|$ is exactly agent $i$'s internal regret until he is selected $M$ times.
	Thus, we obtain
	\begin{align*}
		\eqref{eq:05101204}
		=
		\Pr{\reg_{i}(t) \ge Mx}
		&\le
		\frac{\Ex{\reg_i(t)}}{Mx}.
		\tag{by Markov's inequality}
	\end{align*}
	By Lemma~\ref{th:ETC_main_th}, under the single agent setting when only agent $i$ exists, if we run ETC with $m$ as specified in the theorem statement, we have
	\begin{align*}
		\Ex{\reg_i(t)}  
		\le  
		\sum_{a=1}^{l_i}\Big( \frac{2\delta_{i,a} L^2\ln{(2M)}}{\Delta^2} + 1 \Big).
	\end{align*}
	Thus we have
	\begin{align*}
		\eqref{eq:05101204}
		&\le
		\frac{1}{Mx}	\sum_{a=1}^{l_i}\Big( \frac{2\delta_{i,a} L^2\ln{(2M)}}{\Delta^2} + 1 \Big)
		\\
		&\le
		\frac{1}{Mx}	\sum_{a=1}^{l_i}\Big( \frac{2 L^2\ln{(2M)}}{\Delta^2} + 1 \Big),
	\end{align*}
	and it completes the proof of the claim.
	\end{proof}
	Define $\clean_i$ be the event $\{\muhat_i(M) - \mu_{i,1} \ge \frac{1}{2}\Delta\}$.
	Then, by the claim above, we have
	\begin{align*}
		\Pr{\clean_i^c} 
		&\le 
		\frac{1}{M}\sum_{a=1}^{l_i}\Big( \frac{4 L^2\ln{(2M)}}{\Delta^3} + 1 \Big)
		\\
		&\le
		\frac{1}{M}\sum_{a=1}^{l_i}\Big( \frac{4 L^2\ln{(2M)}}{\Delta^3} + 1 \Big)
	\end{align*}
	This implies that
	\begin{align*}
		\Pr{\cap_{i=1}^n \clean_i} \ge 1 - \frac{1}{M}\sum_{i=1}^n\sum_{a=1}^{l_i}\Big( \frac{4 L^2\ln{(2M)}}{\Delta^3} + 1 \Big).
	\end{align*}
	Note that whenever $\cap_{i=1}^N \clean_i$ happens, the following holds for any $i \in [n]$.
	\begin{align*}
		\muhat_{1}(M)
		&\ge
		\mu_{1,1} - \frac{\Delta}{2}
		\tag{Due to $\clean_1$}
		\\
		&\ge
		\mu_{i,1} + \frac{\Delta}{2}
		\tag{Since $\mu_{1,1} - \mu_{i,1} \ge \Delta$}
		\\
		&\ge\muhat_{i}(M) 
		\tag{Due to $\clean_i$},
	\end{align*}
	and thus agent $1$ will be selected in the Phase 1's exploration phase.
	In this case, the overall regret can be decomposed as follows.
	\begin{align*}
		\reg(T) 
		&\le
		\parans{Mn + \Ex{\reg_i(T)}} + T \cdot \frac{1}{M}\sum_{i=1}^n\sum_{a=1}^{l_i}\Big( \frac{4 L^2\ln{(2M)}}{\Delta^3} + 1 \Big)
		\\
		&\le
		Mn + \sum_{a=1}^{l_i}\Big( \frac{2\delta_{i,a} L^2\ln{(2M)}}{\Delta^2} + 1 \Big) + \frac{T}{M}\sum_{i=1}^n\sum_{a=1}^{l_i}\Big( \frac{4 L^2\ln{(2M)}}{\Delta^3} + 1 \Big)
		\\
		&\le
		Mn + \Big( \frac{2 L^3\ln{(2M)}}{\Delta^2} + l \Big) + \frac{T}{M}\Big( \frac{4 nL^3\ln{(2M)}}{\Delta^3} + nL \Big).
	\end{align*}
	Note that $M = \sqrt{T \ln T}$ optimizes this upper bound (up to constant), and we finally obtain
	\begin{align*}
		\reg(T) 
		&\le \sqrt{T\ln T}n + \frac{L^3 \ln(T\ln T)}{\Delta^2} + \frac{\sqrt{T}}{\sqrt{\ln T}}(\frac{2nL^3\ln(2T \ln T) + \Delta^3nL}{\Delta^3})
		\\
		&=
		O(\frac{nL^3\sqrt{T \ln T}}{\Delta^3}),
	\end{align*}
	and it finishes the proof.
\end{proof}

%% file: icml/icml-ucb.tex
\section{Related Works}\label{app:related}
There exists a growing literature at the intersection of multi-armed bandits and algorithmic game theory. The work of \cite{braverman2019multi} considers a setting where the arms are agents. When an arm (agent) is pulled a reward is realized. The agent can take part of the reward and only report the remaining reward to the principal. 
The work of \cite{feng2020intrinsic} is closely related to \cite{braverman2019multi} in that the agents are still arms that can modify their rewards. However, in \cite{feng2020intrinsic} an agent's utility equals the total number of pulls and are not related to the rewards.
A number of recent works~\cite{esmaeili2023robust,ghosh2013learning,hajiaghayi2023regret} in algorithmic game theory community study the problem of similar spirits, but the problem setups are very different from ours.

Another growing line of works is in the incentivized exploration, initiated by \cite{kremer2014implementing} and \cite{che2018recommender}.
To incentivize a myopic agent arriving at each round and pulling the arm, the principal needs to design an incentive-compatible signaling scheme to induce myopic agents to explore the arms.
To this end, the principal may exploit the information asymmetry~\citep{mansour2015bayesian, immorlica2018incentivizing, sellke2021price} or may compensate the myopic users~\citep{frazier2014incentivizing, chen2018incentivizing}.
Another line of works \citep{banerjee1992simple, bala1998learning, lazer2007network, banihashem2023bandit} has been asking a similar question of \emph{learning dynamics} under a presence of myopic agents pulling the arms, not in a mechanism design perspective, but to analyze the conseqeunce of the given game.

\section{Motivations}\label{app:appl}
Despite being directly motivated by problem posed by~\cite{shin2022multi}, we here present further motivation in terms of the real world applications.
The content replication phenomenon in content platforms like Youtube~\cite{YoutubeReplication} or ad replication in ad platforms such as Google AdSense~\cite{GoogleReplication} would be good practical applications.
In contents platforms like Youtube, a content creator (agent) typically does not exactly know about each content's quality in advance. 
Historical information on the contents can be thought of as a prior distribution of the agent. 
The setting by \cite{shin2022multi} can be interpreted as a scenario in which all the agents know the exact quality of their contents, which is not practical.
Our Bayesian setting indeed captures more realistic scenario of each agent having only a partial information on the expected outcome of each content.
We finally remark that although the problem of replication-proof mechanism has not studied in the bandit literature except~\cite{shin2022multi}, there has been a number of papers for sybil-proof mechanism design in the algorithmic game theory literature~\cite{drucker2012simpler,chen2013sybil}.

\section{Extension to the Asymmetric Setting}\label{app:asym}
Our algorithm and the analysis therein heavily relies on the symmetry assumption such that each agent has a single distribution which is shared over all its arms.
Indeed, our analysis does not directly extend to the asymmetric case. 
In particular, in the asymmetric setting in which each arm is possibly equipped with a different prior distribution, each agent should determine specific number of replication for each arm.
This corresponds to possibly very-different distribution over "bandit instances", which we believe to be significantly challenging to compare in a theoretical manner. 
We firmly believe that such an extension would be very interesting direction from both practical and theoretical perspectives.
However, one can find some evidences on dominant strategies for certain cases in the asymmetric setting. For example, if an agent has two types of arms each sampling from distributions $F_1$ and $F_2$ and if $F_1$ is stochastically dominated by $F_2$, then one can easily argue that $F_2$ is dominant strategy over $F_1$, so it reduces to the symmetric case with $F_2$.

\section{Failure of Existing Algorithms}\label{sec:fail}
We here provide more details on the failure of existing algorithms, presented in Theorem~\ref{thm:negative}.
First, the following is the pseudocode of the standard UCB algoritm in the literature.
\begin{algorithm}
\SetAlgoLined
\textbf{Input}: Tie-breaking rule $\tau$, arm set $A$\\
Let $\muhat_a \gets 0$ and $n_a \gets 0$ for $a \in A$\\
\For{$t=1,2, \ldots$}
{
	\If{$n_{a} = 0$ for some $a \in A$}{Let $a^* \gets a$}
	\Else
	{
		Compute $s_a \gets \muhat_a + \sqrt{2\ln t/n_a}$\\
		Let $\hat{a} \gets \argmax_{a \in A} s_a$\\
	}
}
Pull $\hat{a}$ and obtain reward $R_t$\\
Update statistics: $\muhat_{\hat{a}} \gets \tfrac{n_{\hat{a}} \muhat_{\hat{a}} + R_t}{n_{\hat{a}} + 1}$;~ $n_{\hat{a}} \gets n_{\hat{a}} + 1$
\caption{\textsc{UCB} \cite{auer2002finite}}\label{alg:UCB}
\end{algorithm}

This is extended to the H-UCB, whose pseudocode is presented as follows.
\begin{algorithm} 
\SetAlgoLined
\textbf{Input}: Tie-break rule, agent set $\cN$, arm set $\cS_i$ for $i \in \cN$\\
Let $\muhat_i \gets 0$, $n_i\gets 0$, $\muhat_{i,a} \gets 0$, $n_{i,a} \gets 0$ for $i \in \cN$ and $a \in \cS_i$\\

 \For{$t=1,2,\ldots$}{
  \eIf{$n_i = 0$ for some $i \in \cN$}{ 
    Let $\hat{i} \gets i$
  }{
    Compute $s_i = \muhat_i + \sqrt{2\ln(t)/n_i}$\\
    Let $\hat{i} \gets \argmax_{i \in \cN} s_i$
  }
%
  \eIf{$n_{\hat{i},a} = 0$ for some $a \in \cS_{\hat{i}}$}{ 
    Let $\hat{a} \gets a$
  }{
  	Compute $s_{\hat{i}, a} \gets \muhat_{\hat{i}, a} + \sqrt{2\ln (n_i)/n_{a,t}}$
  }
%
  Pull $\hat{a}$ of agent $\hat{i}$ and obtain reward $R_t$\\
  Update statistics:
  $\muhat_{\hat{i},\hat{a}} \leftarrow \tfrac{\muhat_{\hat{i},\hat{a}}n_{\hat{i}, \hat{a}} + R_t}{n_{\hat{i}, \hat{a}}+1}; ~
  \muhat_{\hat{i}} \leftarrow \tfrac{\muhat_{\hat{i}}n_{\hat{i}} + R_t}{n_{\hat{i}}+1};~ n_{\hat{i}, \hat{a}} \leftarrow n_{\hat{i}, \hat{a}} + 1; ~~ n_{\hat{i}} \leftarrow n_{\hat{i}} + 1;$
 }
 \caption{H-UCB \cite{shin2022multi}\label{alg:hucb}}
\end{algorithm}


As have been noted in the main body of the paper, H-UCB mainly consists of two phases: (i) in the first phase, it runs UCB1 by considering each agent as a single arm and selects one agent at each round, and (ii) in the second phase, it runs UCB within the selected agent.
Precisely, it maintains two types of statistics: (i) agent-wise empirical average reward $\muhat_i$, number of pulls per agent $n_i$, and (ii) arm-wise empirical average reward $\muhat_{i,a}$ and number of pulls per arm $n_{i,a}$.

\cite{shin2022multi} prove that H-UCB admits replication-proofness while maintaining sublinear regret.
Their analysis is largely based on the observation that fully-informed agents \emph{do not} have any incentive to register \emph{suboptimal arms}.
Therefore, their setting essentially reduces to the case under which each agent is equipped with only a \emph{single arm}, and its only choice is to determine how many times to replicate the optimal arm.
This observation largely simplifies the equilibrium analysis since replicating the same arm does not change the expected rewards of the agent due to the hierarchical structure therein.
For example, if one uses uniform-random selection in Phase 1, each agent's expected utility is totally independent from each other.
Thus, it suffices to show truthfulness under single-agent setting.
Especially, if there is only a single agent, \emph{any standard sublinear regret algorithm} would achieve replication-proofness due to the observation above.
Hence, it follows that the pair of uniform-random selection algorithm in the first phase and any standard bandit algorithm in the second phase would achieve replication-proofness.

In our setting, however, the agents only know about the prior distributions from which the arms' mean rewards are drawn, and they are oblivious to which original arm will have the highest mean reward.
This motivates each agent to register multiple arms possibly with some replications. 
This indeed significantly complicates the analysis of equilibrium.
For example, it is not trivial to identify which algorithm is replication-proof even under the \emph{single agent setting}.

In what follows, we decompose Theorem~\ref{thm:negative} by two separate theorems for single-agent and multi-agent arguments, and provide its proofs as well as the corresponding formal restatement.
\begin{theorem}[Single-agent failure - restatement of Theorem~\ref{thm:negative}]\label{thm:ucb_negative}
	There exists a problem instance such that UCB1 is not replication-proof in the single-agent setting.
\end{theorem}
\begin{proof}[Proof of Theorem~\ref{thm:ucb_negative}]
	Consider a single agent with two arms indexed by $1$ and $2$ from distribution $D$ in which each arm's mean reward is $1$ with probability $1/2$ and $0$ with probability $1/2$.
	Consider a truthful strategy $\set{O}$ and strategy $\set{S}$ that replicates arm $1$.
	There exists four possible cases regarding the realization of the arms parameters: (i) both the arms have mean reward $1$, (ii) only arm $1$ has mean reward $1$, (iii) only arm $2$ has mean reward $1$, and (iv) both the arms have mean reward $0$.
	Denote each cases by $X_i$ for $i =1,2,3,4$.
%
	For $X_1$ and $X_4$, note that there exists no difference in the agent's ex-post utility.
	For $X_2$ and $X_3$, we show that there exists a time horizon $T$ such that the average ex-post regret is minimized for strategy $\set{S}$.
	
	\begin{figure}[t]
	\includegraphics[width=0.6\textwidth]{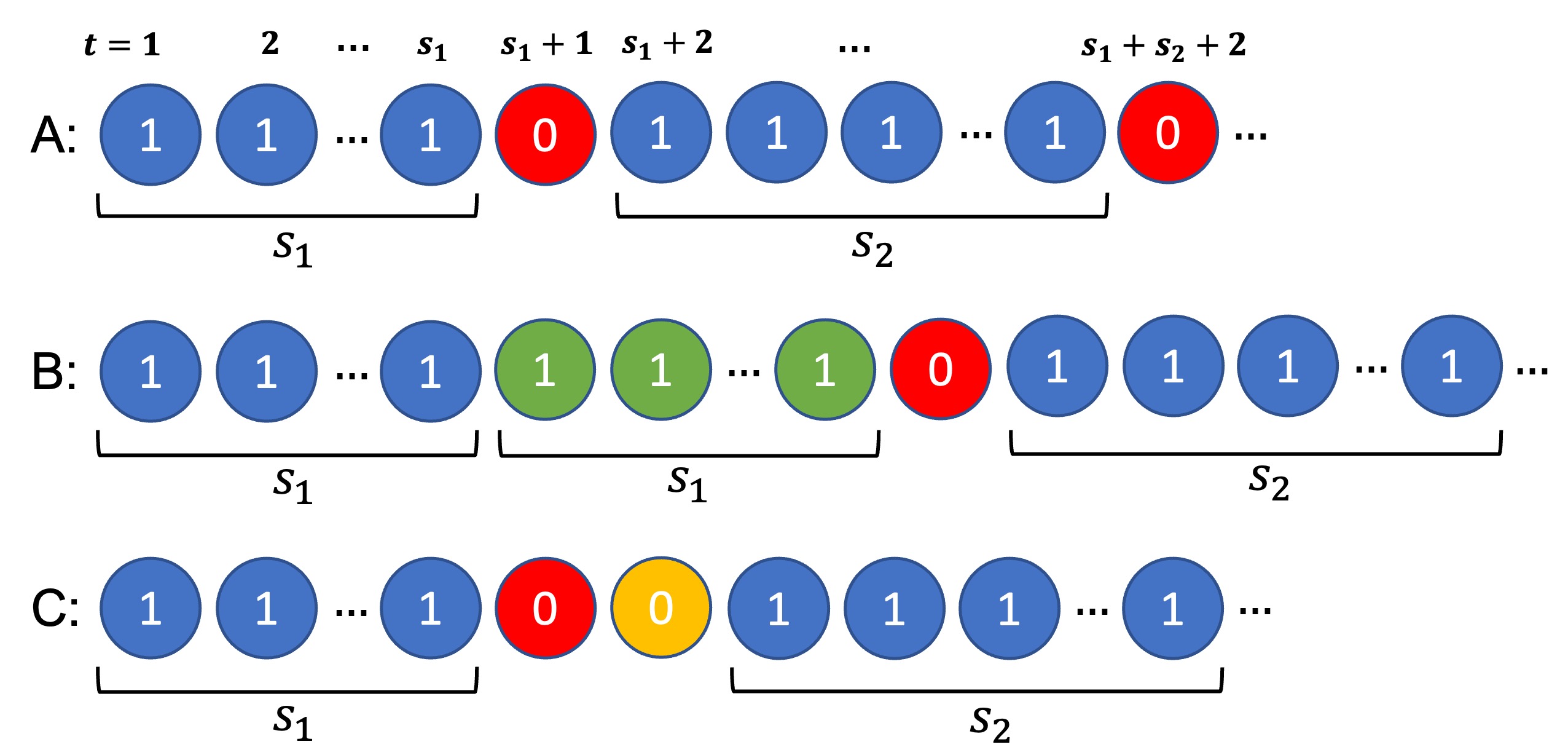}
	\centering
	\caption{Bad instance for UCB in Theorem~\ref{thm:ucb_negative}. 
	Each node denotes the sequence of the bandit algorithm's selection for each instance.
	Blue node denotes the original arm with mean reward $1$, and green one is the replica of it.
	Red nodes the original arm with mean reward $0$, and yellow node is the replica of it.
	Note that at round $t = s_1+s_2+2$, instance $A$ realizes reward $0$ two times, whereas $B$ and $C$ realizes three times, thus having $1.5$ times in average.}
	\label{fig:single-bad}
	\end{figure}
	
	Consider the following list of bandit problem instances with Bernoulli arms.
	\begin{itemize}
		\item Let $A$ be a bandit instance with $\cI_A = (1,0)$.
		\item Let $B$ be a bandit instance with $\cI_B = (1,1',0)$.
		\item Let $C$ be a bandit instance with $\cI_C = (1,0,0')$.
	\end{itemize}
	Each element in $\cI$ denotes the sequence of mean reward of Bernoulli arms, where we write $1'$ and $0'$ to distinguish the arms with same mean.
	Importantly, observe that $A$ denotes the ex-post problem instance that the algorithm faces against strategy $\set{O}$ under $X_2$ or $X_3$.
	Furthermore, $B$ denotes the ex-post problem instance against $\set{S}$ under $X_2$, and $C$ denotes that against $\set{S}$ under $X_3$.
	
	We now show the following claim.
	\begin{claim}\label{cl:neg_ucb}
		Suppose that we run UCB with exploration parameter $f(\cdot)$ for all the problem instances $A,B$, and $C$.
		Then, there exists a time horizon $T$ such that $\reg_A(T) > \nicefrac{(\reg_B(T) + \reg_C(T))}{2}$.
	\end{claim}
	\begin{proof}[Proof of the claim]	
	Define $s_1 = 1$.
	For any integer $i \ge 2$, let $s_i$ denote the minimum integer satisfying
	\begin{align*}
		1 + \sqrt{\frac{\ln(s_i + i)}{s_i}} 
		< 0 + \sqrt{\frac{\ln(s_i + i)}{i}}.
	\end{align*}
	Such $s_i$ exists for any $i \ge 2$ since LHS is a decreasing function of $s_i$ which goes to $1$ as $s_i$ increases and RHS is an increasing function of $s_i$ which goes to $\infty$ as $s_i$ increases, and putting $s_i = i$ yields that LHS is larger than RHS.
	Further, it is easy to check that $s_i > s_j$ if $i > j$ and $s_i \ge i$.
	
	Due to our construction of the sequence $s_i$, in $A$, arm with mean $1$ is selected exactly $s_i$ times and then arm with mean $0$ is selected for a single round, and it repeats for $i=1,2,\ldots$.
	Figure~\ref{fig:single-bad} illustrates this bandit instances.
	As can be seen in the figure, the bandit process in $A$ can be described as
	\begin{align*}
		(1:s_1, 0:1, 1:s_2, 0:1, 1:s_3, 0:1, \ldots),
	\end{align*}
	where in each element of $x:y$ denotes that arm with mean $x$ is selected $y$ times.
	Likewise, in $B$, we have
	\begin{align*}
		(1:s_1, 1':s_1, 0:1, [1:s_2, 1':s_2], 0:1, [1:s_3, 1':s_3], 0:1, \ldots),
	\end{align*}
	where for each element in $[\cdot]$, we ignore the orders how $1$ and $1'$ are selected, but aggregate them to only present the total number of times each one is selected before getting to the arm $0$.
	For $C$, we have
	\begin{align*}
		(1:s_1, [0:1, 0':1], 1:s_2, [0:1, 0':1], 1:s_3, [0:1, 0':1], \ldots).
	\end{align*}
	Set $T = s_1 + 2 + s_2$.
	Given $T$, in $A$, the arm $0$ is selected twice.
	In $B$, the arm $0$ is selected once.
	In $C$, the arm $0$ is selected once, and also $0'$ is selected once.
	Hence, at round $t=T$, $\reg_A(T) = 2 > \frac{1+ 2}{2} = \frac{\reg_B(T) + \reg_C(T)}{2}$,
	and it completes the proof.
	\end{proof}
	Set $T$ as defined in the claim above.
	Finally, we can decompose the agent's ex-ante utility as follows.
	\begin{align*}
		u(\set{O}) 
		= 
		\frac{1}{4}\parans{\sum_{i=1}^4U(\set{O};E_i)}
		&=
		\frac{1}{4}\parans{\sum_{i=1,4}U(\set{O};E_i)} + \frac{1}{4}\parans{\sum_{i=2,3}U(\set{S};E_i)}
		\\
		&=
		\frac{1}{4}\parans{\sum_{i=1,4} \parans{\max_{a \in E_i}\mu_a \cdot T - \reg_{E_i}(T)}}
		+
		\frac{1}{4}\parans{\sum_{i=2,3} \parans{\max_{a \in E_i}\mu_a \cdot T - \reg_{E_i}(T)}}
		\\
		&=
            \frac{T - \reg_A(T)}{2}
		\\
		&<
		\frac{T - (\reg_B(T) + \reg_C(T))/2}{2}
		\\
		&=
		\frac{1}{4}\parans{\sum_{i=1,4}U(\set{S};E_i)} + \frac{1}{4}\parans{\sum_{i=2,3}U(\set{S};E_i)}
		\\
		&= u(\set{S}),
	\end{align*}
	and thus $\set{S}$ yields the strictly larger ex-ante utility, and it concludes the proof.
\end{proof}

Furthermore, even in the multi-agent setting, we prove that H-UCB cannot be made replication-proof with any choice of algorithm parameters, thereby showing that previously suggested algorithms cannot be made replication-proof for any number of agents.
\begin{theorem}[Multi-agent failure - restatement of Theorem~\ref{thm:negative}]\label{thm:hucb_negative}
	For any number of agents $n \ge 1$, there exists a problem instance such that H-UCB is not replication-proof.
\end{theorem}

\begin{discussion}
	We further note that the above two theorems can be easily generalized by parameterizing the exploration parameters $f(n_i, t) = \sqrt{2 \ln(t) / n_i}$ or $g(n_{a,t}, n_i) = \sqrt{2 \ln(n_i)/n_{a,t}}$.
	For example, consider a class of exploration functions $f$ such that it is a monotone decreasing on $n_a$ given $t$, and a monotone increasing on $t$ given $n_a$.
	Further, assume that $f(\cdot, t) \to \infty$ as $t \to \infty$ and $f(n_a, n_a+c) \to 0$ as $n_a \to \infty$ for any constant $c \in \N$, \ie it decreases more rapidly on $n_a$ than it increases from $t$.
	Then, one can check that the both results carry over.
\end{discussion}

\begin{proof}[Proof of Theorem~\ref{thm:hucb_negative}]
	If there exists only a single agent, it is easy to see that the dynamics of the algorithm reduces to standard UCB1.
	Thus, by Theorem~\ref{thm:ucb_negative}, H-UCB is not replication-proof.
	Now we consider multiple agents with $n \ge 2$.
	We only show for the case $n=2$ since it easily extends to the cases with $n \ge 3$.
	Without loss of generality, we assume that the algorithm breaks tie uniformly at random in both phases $1$ and $2$.
	Suppose that agent $1$ has two deterministic arms sampled from $D$ in which the mean reward is $c \in [0,1]$ with probability $0.5$, and $0$ otherwise, for some $c>0$.
	We further assume that agent $2$ is equipped with a deterministic prior $D_i$ supported on $\mu$, and the arm also deterministically outputs reward $\mu$.
	Our choice of $\mu$ will be determined later, which will be selected in order to propagate the single-agent failure to the two-agent cases.
	Especially, we will focus on the round that agent $1$ is selected exactly $s_1+s_2+2$ times, as defined in the proof of Theorem~\ref{thm:ucb_negative}.
	Note that the bandit process is deterministic for any realization of the arms' parameters.
	
	We consider a truthful strategy $\cO$ of registering both arms, and another strategy $\cS$ of replicating only a single arm, and focus on the difference on the agent's expected utility between two strategies.
	If the agent samples $c$ or $0$ both the arms, then the rewards of the bandit process and the agent's ex-post utility therein will be indifferent.
	Thus, we only focus on the remaining case.
	Suppose that agent $1$ samples $c$ for one arm and $0$ for the other, which occurs with probability $1/2$.
	Similar to the proof of Claim~\ref{cl:neg_ucb}, we consnider the following list of bandit problem instances with deterministic arms.
	\begin{itemize}
		\item Let $A$ be a bandit instance with $\cI_A = (c,0)$.
		\item Let $B$ be a bandit instance with $\cI_B = (c,c',0)$.
		\item Let $C$ be a bandit instance with $\cI_C = (c,0,0')$.
	\end{itemize}
	Each element in $\cI$ denotes the sequence of mean reward of Bernoulli arms, where we write $c'$ and $0'$ to distinguish the arms with same mean.
	We simply say arm $c$ and arm $0$ to denote the arm with the corresponding parameter $c$ and $0$, respectively.
	Then, the realization of the agent $1$'s arms in playing strategy $\cS$ corresponds to either of bandit instance $B$ or $C$, and strategy $\cO$ corresponds to the bandit instance $A$.
	Let us abuse the notation of $s_1,s_2,\ldots$ defined in the proof of Claim~\ref{cl:neg_ucb}, so that it denotes the number of rounds that arm $c$ is selected consecutively before arm $0$ is sampled $i$-th time, in problem instance $A$.
	Also, we use $R_t(i)$ and $N_t(i)$ to denote the quantity of $R(i)$ and $N(i)$ at round $t$.
	
	Given the random bits of the tie-breaking rule used by the algorithm, let $\tau^{i}_1,\tau^i_2,\ldots, \tau^i_{k}, \ldots$ be the rounds that agent $1$ is sampled $k$ times when the mean rewards of the agent $1$'s arms are realized to be $i= A,B,C$.
	It is straightforward to check that $\tau^{B}_t \le \tau^{A}_t \le \tau^{C}_t$\footnote{This can  be shown using the coupling technique introduced in \cite{shin2022multi}.}.
	Let $k = s_1+s_2+2$, and following the argument of Claim~\ref{cl:neg_ucb}, we observe that at round $\tau^{B}_k$, agent $1$'s empirical mean rewards by playing $\cS$ will be either of (i) $\frac{c(s_1+s_2)}{s_1+s_2+2}$ or (ii) $\frac{c(s_1+s_2+1)}{s_1+s_2+2}$ based on whether the arms are realized to be instance $B$ or $C$.
	By playing $\cO$, agent $1$ simply has empirical mean reward of $\frac{c(s_1+s_2)}{s_1+s_2+2}$ at round $\tau^{A}_k$.
	Furthermore, if $\tau^{B}_k < \tau^{A}_k$, then in problem instance $A$, it means that the agents' parameters at round $\tau^{B}_k$ satisfy
	\begin{align*}
		R_{\tau^B_k}(1) + \sqrt{\frac{2\ln(t)}{k-1}}
		<
		R_{\tau^B_k}(2) + \sqrt{\frac{2\ln(t)}{\tau^B_k - k}},
	\end{align*}
	where $R_t(i)$ denotes the empirical average reward of agent $i$ at round $t$,  $R_{\tau^B_k}(1) = \frac{c(s_1+s_2)}{s_1+s_2+2}$ and $R_{\tau^B_k}(2) = \mu$.
	However, from our construction of the variables and the fact that H-UCB selects agent $1$ at round $\tau^{B}_k$ in $B$, we know that
	\begin{align*}
		\frac{c(s_1+s_2)}{s_1+s_2+2} + \sqrt{\frac{2\ln(t)}{k-1}}
		\ge
		\mu + \sqrt{\frac{2\ln(t)}{\tau^B_k - k}},
	\end{align*}
	and thus $\tau^A_k$ and $\tau^B_k$ needs to be exactly equivalent.
	Similarly, for instance $C$, since arm $c$ and $c'$ are selected the same number of times at the beginning of round $\tau^B_k$, we conclude that $\tau^i_k$ are equivalent for $i=A,B,C$.
	
	Now define $\theta = \sqrt{\frac{2\ln(t)}{k}} - \sqrt{\frac{2\ln(t)}{\tau^B_k-k}}$ and we pick any $\mu \in (\frac{c(s_1+s_2)}{s_1+s_2+2} + \theta,\frac{c(s_1+s_2+1)}{s_1+s_2+2} + \theta)$.
%
%
%
%
%
%

	Then, at round $\tau^B_k+1$, for bandit instance $A$ and $C$, we observe that agent $2$ will be selected since
	\begin{align*}
		R_{\tau^B_k+1}(2) + \sqrt{\frac{2\ln(t)}{N_{\tau^B_k}(2)}}
		&=
		\mu + \sqrt{\frac{2\ln(t)}{\tau^B_k - k}}
		\\
		&>
		\frac{s_1+s_2}{s_1+s_2+2} + \theta + \sqrt{\frac{2\ln(t)}{\tau^B_k - k}}
		\\
		&=
		\frac{s_1+s_2}{s_1+s_2+2} + \sqrt{\frac{2\ln(t)}{k}}
		= 
		R_{\tau^B_k+1}(1) + \sqrt{\frac{2\ln(t)}{N_{\tau^B_k}(1)}}.
	\end{align*}
	For bandit instance $C$, however, observe that
	\begin{align*}
		R_{\tau^C_k+1}(2) + \sqrt{\frac{2\ln(t)}{N_{\tau^C_k}(2)}}
		&=
		\mu + \sqrt{\frac{2\ln(t)}{\tau^C_k - k}}
		\\
		&<
		\frac{s_1+s_2+1}{s_1+s_2+2} + \theta + \sqrt{\frac{2\ln(t)}{\tau^C_k - k}}
		\\
		&=
		\frac{s_1+s_2+1}{s_1+s_2+2} + \sqrt{\frac{2\ln(t)}{k}} 
		=
		R_{\tau^C_k+1}(1) + \sqrt{\frac{2\ln(t)}{N_{\tau^C_k}(1)}},
	\end{align*}
	and thus arm $1$ is selected.
	This essentially implies that by setting $T = \tau_k^A+1$, agent $1$ is selected at round $T$ with probability $1/4$ if he plays $\cS$, is never selected if he plays $\cO$, where the probability $1/4$ comes from the fact that problem instance $C$ is realized with probability $1/4$.
	Thus, agent $1$'s expected utility becomes larger for such $T$ by replicating one of its arm, and we finish the proof.
\end{proof}